\documentclass[12pt,a4paper]{article}

\usepackage[utf8]{inputenc}
\usepackage[margin=1in]{geometry}
\usepackage{amsmath, amsfonts, amssymb, amsthm}
\usepackage{hyperref}
\usepackage{authblk}
\usepackage[doi=false, isbn=false, url=false, eprint=false]{biblatex}
\renewbibmacro{in:}{%
  \ifentrytype{article}{}{\printtext{\bibstring{in}\intitlepunct}}}
\numberwithin{equation}{section}
\newtheorem{thm}{Theorem}[section]
\newtheorem{lem}[thm]{Lemma}
\newtheorem{prop}[thm]{Proposition}
\theoremstyle{definition}
\newtheorem{dfn}{Definition}[section]
\theoremstyle{remark}
\newtheorem*{rmk}{Remark}

\DeclareMathOperator{\Li}{Li}
\DeclareMathOperator{\D}{d}
\def\CC{\mathbb{C}}

\def\dpa{k} 
\def\dpA{K} 
\def\dpb{k} 
\def\dpB{K} 
\def\hmc{y} 
\def\hmm{m} 
\def\hmM{M} 
\def\ahm{a} 
\def\bhm{b} 
\def\chm{c} 

\title{Legendre transforms for type $A_{n}$ and $B_{n}$ $\vee$-systems}
\author{Misha Feigin}
\author{Leo Kaminski}
\author{Ian A.B. Strachan}
\affil{School of Mathematics and Statistics, University of Glasgow, University Place,
Glasgow G12 8QQ, UK}
\date{}

\begin{document}

\maketitle
\begin{abstract}
    The Witten--Dijkgraaf--Verlinde--Verlinde (WDVV) equations have a rich structure related to the theory of Frobenius manifolds, with many known families of solutions.
    A Legendre transformation is a symmetry of the WDVV equations, introduced by Dubrovin.
    We explicitly compute the results of a Legendre transformation applied to $A_n$- and $B_n$-type multi-parameter rational solutions, relating them to known and new trigonometric solutions.
\end{abstract}

\section{Introduction}\label{sn.intro}
The generalised WDVV equations are the following system of nonlinear partial differential equations:
\begin{equation}\label{eqn.intro.WDVV}
    F_{i}\eta^{-1}F_{j} = F_{j}\eta^{-1}F_{i},
\end{equation}
where $F=F(x^{1},\dots,x^{n})$ and $F_{i}$ is the $n \times n$ matrix of third-order derivatives of the function $F$, with $\left(r,s\right)^{th}$ entry
\begin{equation*}
    F_{irs} := \frac{\partial^{3} F}{\partial x^{i} \partial x^{r} \partial x^{s}} .
\end{equation*}
The $n\times n$ matrix $\eta = \left(\eta_{\alpha\beta}\right)$ is a linear combination
\begin{equation*}
    \eta = \sum^{n}_{k=1} q_{k}F_{k} ,
\end{equation*}
which is assumed to be constant and non-degenerate. In the Frobenius manifold case, we set $q_{k} = \delta_{1k}$ where 1 denotes the identity direction \cite{D93-TFT}; in general, $q_{k}=q_{k}\left(x^{1}, \dots, x^{n}\right)$. 
While $\eta$ is the matrix of a non-degenerate complex bilinear form, we refer to it as a metric. It can be used to lower indices for the coordinates: that is, we define
\begin{equation*}
        x_{\alpha} = \eta_{\alpha\beta}x^{\beta} ,
\end{equation*}
where we assume summation over repeated indices.

A symmetry of the WDVV equations is a map
\begin{align*}
    x^{\alpha} & \mapsto \hat{x}^{\alpha} , \\
    \eta_{\alpha\beta} & \mapsto \hat{\eta}_{\alpha\beta} , \\
    F & \mapsto \widehat{F}=\widehat{F}\left(\hat{x}^{1}, \dots, \hat{x}^{n}\right) ,
\end{align*}
such that the WDVV equations (\ref{eqn.intro.WDVV}) are preserved. That is,
\begin{equation}\label{eqn.intro.WDVVtransformed}
    \widehat{F}_{i}\hat{\eta}^{-1}\widehat{F}_{j} = \widehat{F}_{j}\hat{\eta}^{-1}\widehat{F}_{i}
\end{equation}
provided that $F$ satisfies (\ref{eqn.intro.WDVV}).
The $(r,s)^{th}$ entry of the $n\times n$ matrix $\widehat{F}_{i}$ is
\begin{equation*}
    \widehat{F}_{irs} := \frac{\partial^{3}\widehat{F}}{\partial \hat{x}^{i} \partial \hat{x}^{r} \partial \hat{x}^{s}} .
\end{equation*}
We will use the following notation for partial derivatives:
\begin{align*}
    \partial_{\alpha} & := \frac{\partial}{\partial x^{\alpha}}, \; \hat{\partial}_{\alpha} := \frac{\partial}{\partial \hat{x}^{\alpha}} , \\
    F_{\alpha\beta} & := \partial_{\alpha}\partial_{\beta} F = \frac{\partial^{2} F}{\partial x^{\alpha} \partial x^{\beta} } , \\
    \widehat{F}_{\alpha\beta} & := \hat{\partial}_{\alpha}\hat{\partial}_{\beta} \widehat{F} = \frac{\partial^{2} \widehat{F}}{\partial \hat{x}^{\alpha} \partial \hat{x}^{\beta} } .
\end{align*}
Note that adding quadratic terms to functions $F$ and $\widehat{F}$ does not affect the WDVV equations.

An important example of a symmetry was defined by Dubrovin.
\begin{dfn}\label{dfn.lt} \cite{D93-TFT}
    The Legendre transformation $S_{\gamma}$ is defined as follows for a choice of $\gamma \in \{1, \dots , n\}$:
    \begin{align}
        \hat{x}_{\alpha} &= \partial_{\alpha}\partial_{\gamma}F(x) , \label{eqn.intro.hatx(x)}\\
        \hat{\eta}_{\alpha\beta}(\hat{x}) &= \eta_{\alpha\beta}(x) , \nonumber \\
        \widehat{F}_{\alpha\beta}(\hat{x}) &= F_{\alpha\beta}(x) , \nonumber
    \end{align}
    where $\hat{x}_{\alpha}=\hat{\eta}_{\alpha\beta}\hat{x}^{\beta}$.
\end{dfn}
\noindent A consequence of this definition is that the components of the metric $\hat{\eta}$ can be written as
\begin{equation*}
    \hat{\eta}_{\alpha\beta}\left(\hat{x}\right) = \eta_{\alpha\beta}(x) = \sum_{k}q_{k}\partial_{k}F_{\alpha\beta} = \sum_{k, l} q_{k} \frac{\partial\hat{x}^{l}}{\partial x^{k}} \frac{\partial}{\partial \hat{x}^{l}}\widehat{F}_{\alpha\beta} = \sum_{l}\hat{q}_{l}\widehat{F}_{\alpha\beta l} ,
\end{equation*}
for $\hat{q}_{l} = \sum_{k}q_{k}\frac{\partial\hat{x}^{l}}{\partial x^{k}}$ .

\begin{prop} \cite{D93-TFT}, \cite{SS17-GLT}. 
The Legendre transformed prepotential $\widehat{F}=S_{\gamma}\left(F\right)$ and metric $\hat{\eta}$, given by formulas in Definition \ref{dfn.lt}, satisfy the WDVV equations (\ref{eqn.intro.WDVVtransformed}) provided that $F$ satisfies the WDVV equations (\ref{eqn.intro.WDVV}).
\end{prop}
\noindent  We denote $\widehat{F} = S_{\gamma}\left(F\right)$, although $\widehat{F}$ depends on both $F$ and $\eta$.
From the computational point of view, difficulties in finding $\widehat{F}$ arise when inverting expressions of the form $\hat{x}^{\alpha}(x)$, obtained from (\ref{eqn.intro.hatx(x)}). In general, it is not always possible to carry out this inversion.

We are interested in the Legendre transformations of rational solutions of the WDVV equations. Such solutions are defined by a collection of vectors known as $\vee$-systems introduced by Veselov \cite{V99-VS}, which extend the class of Coxeter root systems and can be characterised by geometric properties. 
In the Coxeter case, these solutions appear as almost dual prepotentials for the orbit space Frobenius manifolds $\CC^{n}/W$, for a finite Coxeter group $W$ \cite{D04-AD}.
Chalykh and Veselov found rational $\vee$-systems which are multi-parameter deformations of root systems $A_{n}$ and $B_{n}$ \cite{CV01-VS}; while such deformations produce solutions to the WDVV equations, they cannot in general be related to a Frobenius manifold by almost duality. We find that Legendre transformations of these two families of rational solutions produce trigonometric solutions of the WDVV equations.

Trigonometric solutions have been constructed as almost dual Frobenius manifold prepotentials on the orbit spaces of extended affine Weyl groups - see \cite{RS07-RTLT} for type $A_{n}$, and were obtained via the (equivariant) Gromov-Witten generating functions for resolutions of $\CC^{2}/G$ \cite{BG08-GW}.
Rational and trigonometric solutions also appear in $N = 2$ supersymmetric gauge theories; Marshakov, Mironov and Morosov showed that perturbative parts of some prepotentials arising in Seiberg-Witten field theory are of this type \cite{MM00-WDVV}, and Hoevenaars and Martini studied trigonometric solutions related to classical root systems \cite{HM03-5D}.
There exist multi-parameter families of trigonometric solutions related to deformations of classical root systems $A_{n}$ and $BC_{n}$, which were found by Pavlov \cite{P06-EV}, generalised by Alkadhem and Feigin \cite{AF20-TRIGVS}, and linked to flat connections by Shen \cite{S19-TRS} in type $A_n$; see also Riley \cite{R07-PHD} for other multi-parameter $A_n$-type solutions. A large class of trigonometric solutions was considered by Feigin in \cite{F09-TVS}.

Consider the following simple example, a rational solution which corresponds to the Weyl group $A_{2}$:
\begin{equation}\label{eqn.intro.eg1}
    F = \left(x^{1}\right)^{2}\log{x^{1}} + \left(x^{2}\right)^{2}\log{x^{2}} + \left(x^{1}-x^{2}\right)^{2}\log\left(x^{1}-x^{2}\right) .
\end{equation}
Here, the metric is given by
\begin{equation*}
    \eta =  x^{1}F_{1}+x^{2}F_{2} = 
    \begin{pmatrix}
        4 & -2 \\
        -2 & 4
    \end{pmatrix} .
\end{equation*}
Applying the Legendre transformation $S_{1}$ to $F$ produces the solution
\begin{equation}\label{eqn.intro.eg}
    \widehat{F} = S_{1}(F) = \frac{2}{3}\left(\hat{x}^{1}\right)^{3} + \left(\hat{x}^{1}\right)^{2}\hat{x}^{2} + 2\,\hat{x}^{1}\left(\hat{x}^{2}\right)^{2} -\frac{1}{3}\left(\hat{x}^{2}\right)^{3} - \frac{2}{9}\Li_{3}\left(e^{-3\hat{x}^{2}}\right) ,
\end{equation}
where $\Li_{3}$ is the trilogarithm function.
Solution (\ref{eqn.intro.eg}) is an example of a trigonometric solution of the WDVV equations which lies outside of the class considered in \cite{F09-TVS}.

 An early example of rational solutions being mapped to trigonometric solutions by a Legendre transformation is given by Riley and Strachan \cite{RS07-RTLT}. Here, the rational solutions are almost dual to Hurwitz space Frobenius manifolds of rational functions on the sphere with a single finite pole; this pole provides a natural choice of Legendre direction.
A type $B$ version of this construction is discussed by Stedman and Strachan in \cite{SS21-EVS} at the level of superpotentials. Further examples of Legendre transformations are given in \cite{Y23-ALT}, where it was shown that monodromy data is the same for a pair of Frobenius manifolds related by a Legendre transformation.

Note that the solution in (\ref{eqn.intro.eg1}) does not have a distinguished direction associated with it, unlike the settings  considered in \cite{RS07-RTLT}, \cite{SS21-EVS}. Thus in this example we have chosen the coordinate direction $x^{1}$ as the Legendre direction. In the general case, discussed in Sections \ref{sn.Anrat} and \ref{sn.Bnrat}, we find the results of Legendre transformations in an arbitrary coordinate direction, $x^{\gamma}$ for $\gamma \in \{1, \dots, n\}$, applied to multi-parameter $A_n$- and $B_n$-type rational $\vee$-systems in $n$ dimensions. For both of these families, we find that the solution produced by such a Legendre transformation is trigonometric; in the $B_{n}$-type case, we show that the Legendre-transformed solution coincides with the family of $BC_{n}$-type trigonometric solutions found in \cite{AF20-TRIGVS}, \cite{P06-EV}.

In the case of $A_{n}$-type trigonometric solutions, we generalise known solutions from Hoevenaars and Martini \cite{HM03-5D} to a multi-parameter family of solutions. The initial solution (which appears only in the arXiv preprint version of \cite{HM03-5D}) was also considered by Shen \cite{S19-TRS}, where it was presented in a different but equivalent form. The multi-parameter family which we obtain also generalises the multi-parameter solutions from  \cite{P06-EV}, \cite{R07-PHD}, \cite{AF20-TRIGVS}. We find that this new larger class of solutions contains Legendre transformations of the general $A_{n}$-type rational solutions.

The layout of the paper is as follows. In Section \ref{sn.Anrat}, we compute the general formula for the Legendre transformation of the $A_n$-type multi-parameter rational solutions. In Section \ref{sn.Bnrat}, we do the same for the $B_n$-type rational solutions. Section \ref{sn.Antrig} introduces the new family of trigonometric $A_n$-type solutions. Section \ref{sn.rat2trig} explains the relations between the Legendre transformations of $A_n$- and $B_n$-type rational solutions, and the respective families of trigonometric solutions.

\section{Legendre transform of \texorpdfstring{$A_{n}$}{An}-type solutions}\label{sn.Anrat}
The family of rational $A_{n}$-type WDVV solutions corresponding to multi-parameter $A_{n}$-type $\vee$-systems \cite{CV01-VS} is given by
\begin{equation} \label{eqn: An rat sol}
    F = F^{\textrm{rat}}_{A_{n}(\dpa)}=\sum^{n}_{i=1}\dpa_{i}\left(x^{i}\right)^2\log\left(x^{i}\right)+\sum_{1\leq i<j\leq n}\dpa_{i}\dpa_{j}\left(x^{i}-x^{j}\right)^{2}\log\left(x^{i}-x^{j}\right) ,
\end{equation}
with parameters $\dpa_{1}, \dots, \dpa_{n} \in \CC^{\times}$ such that $\dpA := \sum^{n}_{i=1} \dpa_{i} \neq -1 .$

For the generalised WDVV equations (\ref{eqn.intro.WDVV}), it is natural to choose $\eta$ as follows \cite{FV08-VS}:
\begin{equation}\label{eqn: Vmetric}
    \eta_{\alpha\beta} = \sum^{n}_{i=1}x^{i}F_{i\alpha\beta}.
\end{equation}
This choice of $\eta$ is a flat metric with the flat coordinate system $x^{\alpha}$. It has constant entries
\begin{equation*}
    \eta_{\alpha\beta} = 
    \begin{cases}
        2\dpa_{\alpha}(\dpA-\dpa_{\alpha}+1) & \textrm{if } \alpha = \beta , \\
        -2\dpa_{\alpha}\dpa_{\beta} & \textrm{otherwise.}
    \end{cases}
\end{equation*}
Then the inverse metric, $\eta^{-1}$, has entries given by
 \begin{equation}\label{eqn.Anrat.invmetric}
     \eta^{\alpha\beta} =
     \begin{cases}
         \frac{\dpa_{\alpha}+1}{2\dpa_{\alpha}(\dpA+1)} & \textrm{ if } \alpha = \beta , \\
         \frac{1}{2(\dpA+1)} & \textrm{ otherwise.}
     \end{cases}
 \end{equation}
We now compute the result of applying a Legendre transformation $S_{\gamma}$ to $F$ for any $\gamma \in \{1, \dots, n\}$.
Using Definition \ref{dfn.lt}, we set $\widehat{F}=S_{\gamma}\left(F\right)$.

Note that since the WDVV equations are unaffected by the addition of quadratic or lower-order parts to a function $F$, we may disregard such terms during calculations.
By definition, the new coordinates $\hat{x}_{\alpha}$ are found using second-order derivatives of $F$. For simplicity, constant terms that appear from these calculations can be absorbed with the addition of appropriate quadratic terms to $F$, and the flat coordinates $\hat{x}_{\alpha}$ redefined. Since the addition of such terms to $F$ is equivalent to a coordinate translation of the resulting $\widehat{F}$, we will find and discuss $\widehat{F}$ up to such transformations.
We will use the notation $A \overset{\text{\tiny$\bullet$}}{=} B$ to denote the equality of $A$ and $B$ up to constant terms.

Under the Legendre transformation $S_{\gamma}$, the new flat coordinates $\hat{x}_{\alpha}$ can be chosen (after constant coordinate shifts) as follows:
\begin{equation}\label{eqn: An rat hat cov}
    \hat{x}_{\alpha} = 
    \begin{cases}
        2\dpa_{\gamma} \log(x^{\gamma}) +2\dpa_{\gamma}\sum_{i \neq \gamma}\dpa_{i}\log\left(x^{\gamma} -x^{i}\right) & \textrm{if } \alpha = \gamma , \\
        -2\dpa_{\gamma}\dpa_{\alpha}\log\left(x^{\gamma}-x^{\alpha} \right) & \textrm{otherwise.}
    \end{cases}
\end{equation}
The original flat coordinates $x^{\alpha}$ now must be rewritten in terms of the new contravariant coordinates $\hat{x}^{\alpha}$.
\begin{lem}\label{thm: An cont coord}
    The flat coordinates $x^{\alpha}$ can be expressed as follows:
    \begin{align}
        x^{\alpha} & = \left(1-e^{-\frac{\dpA+1}{\dpa_{\gamma}}\hat{x}^{\alpha}}\right) x^{\gamma} \; \textrm{\emph{for} } \alpha \neq \gamma , \label{eqn: xalph cont}\\
        \log\left(x^{\gamma}\right) & = \sum_{i=1}^{n}\frac{\dpa_{i}}{\dpa_{\gamma}}\hat{x}^{i} \label{eqn: logx} .
    \end{align}
\end{lem}
\begin{proof}
    The new contravariant coordinates $\hat{x}^{\alpha}$ can be found by using the inverse metric $\eta^{-1}$ (\ref{eqn.Anrat.invmetric}) to raise the indices of the covariant coordinates in (\ref{eqn: An rat hat cov}). Since the Legendre transformation does not affect the metric, we have
    \begin{align*}
        & \eta_{\alpha\beta}= \hat{\eta}_{\alpha\beta} , \\
        \intertext{which implies that}
        & \hat{x}^{\alpha}=\hat{\eta}^{\alpha\beta}\hat{x}_{\beta}=\eta^{\alpha\beta}\hat{x}_{\beta} .
    \end{align*}
    The contravariant coordinates are therefore
    \begin{equation} \label{eqn: hat cont}
        \hat{x}^{\alpha} =
        \begin{cases}
            \frac{\dpa_{\gamma}+1}{\dpA+1}\log\left(x^{\gamma}\right) + \frac{1}{\dpA+1}\sum_{i\neq \gamma}\dpa_{i}\log\left(x^{\gamma}-x^{i}\right) & \textrm{if } \alpha = \gamma , \\[0.5em]
            \frac{\dpa_{\gamma}}{\dpA+1}\log\left(x^{\gamma}\right) - \frac{\dpa_{\gamma}}{\dpA+1}\log\left(x^{\gamma}-x^{\alpha} \right) & \textrm{otherwise.}
        \end{cases}
    \end{equation}
    Considering the case when $\alpha \neq \gamma$, we have
    \begin{equation*}
        \hat{x}^{\alpha} = \frac{\dpa_{\gamma}}{\dpA+1}\log\left(\frac{x^{\gamma}}{x^{\gamma}-x^{\alpha}}\right) ,
    \end{equation*}
    which can be rearranged to expression (\ref{eqn: xalph cont}).
    In the case where $\alpha=\gamma$ in (\ref{eqn: hat cont}), we see
    \begin{equation*}
        \hat{x}^{\gamma} = \frac{\dpa_{\gamma}+1}{\dpA+1}\log\left(x^{\gamma}\right) + \frac{1}{\dpA+1}\sum_{i\neq \gamma}\dpa_{i}\log\left(x^{\gamma}-x^{i}\right).
    \end{equation*}
    Using (\ref{eqn: xalph cont}) to substitute for $x^{i}$ in the summation term gives
    \begin{align*}
        \hat{x}^{\gamma} 
        & = \frac{\dpa_{\gamma}+1}{\dpA+1}\log\left(x^{\gamma}\right) + \frac{1}{\dpA+1}\sum_{i \neq \gamma}\dpa_{i}\left(\log\left(x^{\gamma}\right)+\log\left(e^{-\frac{\dpA+1}{\dpa_{\gamma}}\hat{x}^{i}}\right)\right) \\
        & = \log\left(x^{\gamma}\right) - \sum_{i\neq\gamma}\frac{\dpa_{i}}{\dpa_{\gamma}}\hat{x}^{i}.
    \end{align*}
    This is easily rearranged to find (\ref{eqn: logx}).
\end{proof}

We may now find the second-order derivatives of $\widehat{F}$.
\begin{lem}\label{thm: An rat Fab}
    The second-order derivatives $\widehat{F}_{\alpha\beta}$ can be expressed by the following formulae for all $\alpha, \beta \in \{1, \dots, n\}$. \\
    Case 1: $\alpha = \gamma$.
    \begin{equation*}
        \widehat{F}_{\gamma\beta} \overset{\text{\tiny$\bullet$}}{=} 2\dpa_{\beta}(\dpA-\dpa_{\beta}+1)\hat{x}^{\beta}-2\dpa_{\beta}\sum_{i \neq \beta}\dpa_{i}\hat{x}^{i}.
    \end{equation*}
    Case 2: $\alpha = \beta \neq \gamma$.
    \begin{multline*}
        \widehat{F}_{\alpha\alpha} \overset{\text{\tiny$\bullet$}}{=} 2\dpa_{\alpha}(\dpA+1)\hat{x}^{\gamma} + \frac{2\dpa_{\alpha}}{\dpa_{\gamma}}(\dpa_{\alpha}-\dpa_{\gamma})(\dpA+1)\hat{x}^{\alpha} - \frac{2\dpa_{\alpha}^{2}}{\dpa_{\gamma}}\sum_{i=1}^{n}\dpa_{i}\hat{x}^{i} \\
        + 2\dpa_{\alpha}\log\left(1-e^{-\frac{\dpA+1}{\dpa_{\gamma}}\hat{x}^{\alpha}}\right) + 2\dpa_{\alpha}\sum_{i\neq\alpha,\gamma}\dpa_{i}\log\left(1-e^{-\frac{\dpA+1}{\dpa_{\gamma}}\left(\hat{x}^{\alpha}-\hat{x}^{i}\right)}\right) .
    \end{multline*}
    Case 3: $\alpha, \beta, \gamma$ distinct.
    \begin{equation*}
        \widehat{F}_{\alpha\beta} \overset{\text{\tiny$\bullet$}}{=} -\frac{2\dpa_{\alpha}\dpa_{\beta}}{\dpa_{\gamma}}\sum_{i=1}^{n}\dpa_{i}\hat{x}^{i} - 2\dpa_{\alpha}\dpa_{\beta}\log\left(1-e^{-\frac{\dpA+1}{\dpa_{\gamma}}(\hat{x}^{\alpha}-\hat{x}^{\beta})}\right) + \frac{2\dpa_{\alpha}\dpa_{\beta}(\dpA+1)}{\dpa_{\gamma}}\hat{x}^{\beta} .
    \end{equation*}
\end{lem}
\begin{proof}
    Case 1: $\alpha=\gamma$. \newline
    Using previously defined properties of the Legendre transformation $S_{\gamma}$ and the coordinate system $\hat{x}$, we can see that
    \begin{equation*}
        \widehat{F}_{\gamma\beta} = F_{\gamma\beta} \overset{\text{\tiny$\bullet$}}{=} \hat{x}_{\beta} = \hat{\eta}_{\alpha\beta}\hat{x}^{\alpha} = \eta_{\alpha\beta}\hat{x}^{\alpha} ,
    \end{equation*}
    which leads to the statement.
    
    Case 2: $\alpha = \beta \neq \gamma$. \newline
    First we calculate the relevant second-order derivatives of $F$:
    \begin{equation*}
        \widehat{F}_{\alpha\alpha} = F_{\alpha\alpha} \overset{\text{\tiny$\bullet$}}{=} 2\dpa_{\alpha}\log\left(x^{\alpha}\right) + 2\dpa_{\alpha}\sum_{i\neq\alpha}\dpa_{i}\log\left(x^{\alpha}-x^{i}\right).
    \end{equation*}
    Using expression (\ref{eqn: xalph cont}) to make substitutions, we find that
    \begin{multline*}
        \widehat{F}_{\alpha\alpha} \overset{\text{\tiny$\bullet$}}{=} 2\dpa_{\alpha}(\dpA-\dpa_{\alpha}+1)\log\left(x^{\gamma}\right) + 2\dpa_{\alpha}\log\left(1-e^{-\frac{\dpA+1}{\dpa_{\gamma}}\hat{x}^{\alpha}}\right) \\
        - 2\dpa_{\alpha}(\dpA+1)\hat{x}^{\alpha} + 2\dpa_{\alpha}\sum_{i\neq\alpha,\gamma}\dpa_{i}\log\left(e^{-\frac{\dpA+1}{\dpa_{\gamma}}\hat{x}^{i}}-e^{-\frac{\dpA+1}{\dpa_{\gamma}}\hat{x}^{\alpha}}\right).
    \end{multline*}
    We continue as follows, using expression (\ref{eqn: logx}) to substitute for $\log(x^{\gamma})$:
    \begin{multline*}
        \widehat{F}_{\alpha\alpha} \overset{\text{\tiny$\bullet$}}{=} 2\dpa_{\alpha}(\dpA-\dpa_{\alpha}+1)\hat{x}^{\gamma} + \frac{2\dpa_{\alpha}}{\dpa_{\gamma}}\big((\dpa_{\alpha}-\dpa_{\gamma})(\dpA+1)-\dpa_{\alpha}^{2}\big)\hat{x}^{\alpha}\\ 
        + 2\dpa_{\alpha}\log\left(1-e^{-\frac{\dpA+1}{\dpa_{\gamma}}\hat{x}^{\alpha}}\right) + 2\dpa_{\alpha}\sum_{i\neq\alpha,\gamma}\dpa_{i}\log\left(1-e^{-\frac{\dpA+1}{\dpa_{\gamma}}\left(\hat{x}^{\alpha}-\hat{x}^{i}\right)}\right) \\
        + 2\dpa_{\alpha}\sum_{i\neq\alpha,\gamma}\dpa_{i}\log\left(e^{-\frac{\dpA+1}{\dpa_{\gamma}}\hat{x}^{i}}\right) + 2\dpa_{\alpha}(\dpA-\dpa_{\alpha}+1)\sum_{i=1}^{n}\frac{\dpa_{i}}{\dpa_{\gamma}}\hat{x}^{i}.
    \end{multline*}
    With some rearranging, this leads to the required statement.
    
    Case 3: $\alpha, \beta, \gamma$ distinct. \\
    Here we again use the formulae in Lemma \ref{thm: An cont coord} to show that
    \begin{align*}
        \widehat{F}_{\alpha\beta} & = F_{\alpha\beta} \\
        & \overset{\text{\tiny$\bullet$}}{=} - 2\dpa_{\alpha}\dpa_{\beta}\log\left(x^{\alpha}-x^{\beta}\right) \\
        & = -2\dpa_{\alpha}\dpa_{\beta}\log(x^{\gamma}) - 2\dpa_{\alpha}\dpa_{\beta}\log\left(e^{-\frac{\dpA+1}{\dpa_{\gamma}}\hat{x}^{\beta}}-e^{-\frac{\dpA+1}{\dpa_{\gamma}}\hat{x}^{\alpha}}\right) \\
        & = -2\dpa_{\alpha}\dpa_{\beta}\sum_{i=1}^{n}\frac{\dpa_{i}}{\dpa_{\gamma}}\hat{x}^{i} - 2\dpa_{\alpha}\dpa_{\beta}\log\left(1-e^{-\frac{\dpA+1}{\dpa_{\gamma}}(\hat{x}^{\alpha}-\hat{x}^{\beta})}\right) + \frac{2\dpa_{\alpha}\dpa_{\beta}(\dpA+1)}{\dpa_{\gamma}}\hat{x}^{\beta} ,
    \end{align*}
    which has the required form.
\end{proof}

Recall that the polylogarithm $\Li_{n}$ is defined as
 \begin{equation*}
     \Li_{n}(z) = \begin{cases}
        -\log(1-z) &\textrm{ if $n=1$,}\\
        \sum_{k=1}^{\infty}\frac{z^{k}}{k^{n}} &\textrm{ otherwise,}
     \end{cases}
 \end{equation*}
 with derivatives behaving in the following way:
\begin{equation*}
    \frac{\D \Li_{n}\left(e^{u}\right)}{\D u} = \Li_{n-1}(e^{u}) .
\end{equation*}
In this and the following sections, we consider WDVV solutions which can be written in terms of the function
\begin{equation}\label{eqn: f def}
    f(z) = \frac{1}{6}z^{3} - \frac{1}{4} \Li_{3}\left(e^{-2z}\right),
\end{equation}
where $f'''(z) = \coth{z}$.
It is useful to note the second-order derivative
\begin{equation}\label{eqn.Anrat.f''}
    \frac{\text{d}^{2}}{\text{d}z^{2}}\big( 8 f(\kappa z/2)\big) = \kappa^{3} z + 2\kappa^{2} \log\left(1-e^{-\kappa z}\right) ,
\end{equation}
where $\kappa \in \CC$.

\begin{thm}\label{thm: An rat LT}
    The Legendre transform $S_{\gamma}$ of the function $F =F^{\textrm{rat}}_{A_{n}(\dpa)}$, given by (\ref{eqn: An rat sol}), has the form
    \begin{multline}\label{eqn: An rat LT}
        \widehat{F} = \frac{\eta_{\gamma\gamma}}{6}(\hat{x}^{\gamma})^{3} + (\hat{x}^{\gamma})^{2}\sum_{i\neq\gamma}\frac{\eta_{i\gamma}}{2}\hat{x}^{i} + \hat{x}^{\gamma}\sum_{i\neq\gamma}\frac{\eta_{ii}}{2}(\hat{x}^{i})^{2}  - \sum_{i<j<l}\frac{2\dpa_{i}\dpa_{j}\dpa_{l}}{\dpa_{\gamma}}\hat{x}^{i}\hat{x}^{j}\hat{x}^{l} \\
        + \sum_{i\neq\gamma}
        \left(\frac{\dpa_{i}(\eta_{ii}-\eta_{\gamma\gamma})}{12\dpa_{\gamma}}-\frac{\eta_{ii}^{2}}{24\dpa_{i}\dpa_{\gamma}}+\frac{\eta_{i\gamma}}{12}\right)(\hat{x}^{i})^{3} + \sum_{\substack{i \neq j \\ i,j\neq\gamma}}\frac{\dpa_{j}\eta_{ii}+\dpa_{i}\eta_{ij}}{4\dpa_{\gamma}}(\hat{x}^{i})^{2}\hat{x}^{j} \\
        + \sum_{i\neq\gamma}\frac{8\dpa_{\gamma}^{2}\dpa_{i}}{\left(\dpA+1\right)^{2}}f\left(\frac{\dpA+1}{2\dpa_{\gamma}}\hat{x}^{i}\right) + \sum_{\substack{i<j \\ i,j\neq\gamma}}\frac{8\dpa_{\gamma}^{2}\dpa_{i}\dpa_{j}}{\left(\dpA+1\right)^{2}}f\left(\frac{\dpA+1}{2\dpa_{\gamma}}(\hat{x}^{i}-\hat{x}^{j})\right) 
    \end{multline}
    up to quadratic terms and coordinate shifts.
\end{thm}
\begin{proof}
    This can be checked by taking (\ref{eqn: An rat LT}) as an ansatz and comparing its second-order derivatives with the expressions given in Lemma {\ref{thm: An rat Fab}}. In particular, we use (\ref{eqn.Anrat.f''}) to rewrite the logarithmic terms.
\end{proof}
\noindent Comparison of this solution with other trigonometric solutions is done in Section \ref{sn.Anr2t}.

\section{Legendre transform of \texorpdfstring{$B_{n}$}{Bn}-type solutions}\label{sn.Bnrat}
The rational prepotentials of $B_{n}$-type \cite{CV01-VS} are given by
\begin{equation}\label{eqn: Bn rat sol}
    F^{\textrm{rat}}_{B_{n}(\dpb)} = \sum_{i=1}^{n} 2\dpb_{i}(\dpb_{0}+\dpb_{i}) (x^{i})^{2}\log(x^{i}) + \sum_{1 \leq i < j \leq n } \dpb_{i}\dpb_{j} (x^{i} \pm x^{j})^{2}\log(x^{i} \pm x^{j}) ,
\end{equation}
with parameters $\dpb_{0}, \dpb_{1}, \dots, \dpb_{n} \in \CC^{\times}$ and where the sum of the parameters
\begin{equation}\label{eqn: Bn rat dpB}
    \dpB := \sum^{n}_{i=0} \dpb_{i} \neq 0 .
\end{equation}
The metric for this system, $\eta$ as defined in (\ref{eqn: Vmetric}), has entries
\begin{equation*}
    \eta_{\alpha\beta} = 4\dpb_{\alpha}\dpB\delta_{\alpha\beta}
\end{equation*}
while the inverse metric, $\eta^{-1}$, has entries given by
 \begin{equation}\label{eqn.Bnrat.invmet}
     \eta^{\alpha\beta} = \frac{\delta_{\alpha\beta}}{4\dpb_{\alpha}\dpB} .
 \end{equation}

We can now apply the Legendre transformation $S_{\gamma}$, as in Definition \ref{dfn.lt}, to $F=F^{\text{rat}}_{B_{n}(\dpb)}$ for an arbitrary $\gamma \in \{1, \dots, n\}$. Throughout this subsection, we set $\widehat{F}=S_{\gamma}(F)$.
As discussed previously, solutions to the WDVV equations are only defined up to quadratic terms and so we may omit, for example, constant terms arising from second-order derivatives of such solutions.
Using the definition of a Legendre transformation, we may choose the set of new (covariant) coordinates $\hat{x}_{\alpha}$ as follows
\begin{equation}\label{eqn: Bn rat hat cov}
    \hspace{-0.5em} \hat{x}_{\alpha} = 
    \begin{cases}
        4\dpb_{\gamma}(\dpb_{\gamma}+\dpb_{0})\log(x^{\gamma})+2\dpb_{\gamma}\displaystyle{\sum_{\substack{i=1\\i \neq \gamma}}^{n}} \dpb_{i}\log\left((x^{\gamma})^{2}-(x^{i})^{2}\right) & \textrm{if } \alpha = \gamma , \\
        2\dpb_{\gamma}\dpb_{i}\log\left(\frac{x^{\gamma}+x^{i}}{x^{\gamma}-x^{i}}\right) & \textrm{otherwise.}
    \end{cases}
\end{equation}
To continue applying the transformation, we now must find expressions for coordinates $x^{\alpha}$ in terms of $\hat{x}^{\alpha}$.

\begin{lem}\label{thm: Bn rat xinv}
    The flat coordinates $x^{\alpha}$ can be expressed as follows:
    \begin{align}
        x^{\alpha} & = x^{\gamma}\coth\left(\frac{\dpB}{\dpb_{\gamma}}\hat{x}^{\alpha}\right) \textrm{ \emph{for} } \alpha \neq \gamma, \label{eqn: Bn rat xinv1}\\
        \log(x^{\gamma}) & = \hat{x}^{\gamma} - \sum_{\substack{i=1\\i \neq \gamma}}^{n} \frac{\dpb_{i}}{2\dpB} \log\left(1-\coth^{2}\left(\frac{\dpB}{\dpb_{\gamma}}\hat{x}^{i}\right)\right). \label{eqn: Bn rat xinv2}
    \end{align}
\end{lem}
\begin{proof}
    First, we use the inverse metric as in (\ref{eqn.Bnrat.invmet}) to raise the indices of the covariant coordinates $\hat{x}_{\alpha}$ in (\ref{eqn: Bn rat hat cov}) via the formula
    \begin{equation*}
        \hat{x}^{\alpha} = \hat{\eta}^{\alpha\beta}\hat{x}_{\beta} = \eta^{\alpha\beta}\hat{x}_{\beta} .
    \end{equation*}
    We obtain
    \begin{equation}\label{eqn: Bn rat hat cont}
        \hat{x}^{\alpha} = 
        \begin{cases}
            \frac{\dpb_{0}+\dpb_{\gamma}}{\dpB}\log(x^{\gamma}) + \sum_{i \neq 0,\gamma} \frac{\dpb_{i}}{2\dpB}\log\left(\left(x^{\gamma}\right)^{2}-(x^{i})^{2}\right) &\textrm{if } \alpha=\gamma, \\[0.5em]
            \frac{\dpb_{\gamma}}{2\dpB}\log\left(\frac{x^{\gamma}+x^{\alpha}}{x^{\gamma}-x^{\alpha}}\right) &\textrm{otherwise.}
        \end{cases}
    \end{equation}
    
    In the case where $\alpha\neq\gamma$, we can invert the expression for $\hat{x}^{\alpha}$ to obtain
    \begin{equation*}
        e^{\frac{2\dpB}{\dpb_{\gamma}}\hat{x}^{\alpha}} = \frac{x^{\gamma}+x^{\alpha}}{x^{\gamma}-x^{\alpha}},
    \end{equation*}
    which can be rearranged to find
    \begin{equation*}
        x^{\alpha} = x^{\gamma}\left(\frac{e^{2\dpB\hat{x}^{\alpha}/\dpb_{\gamma}}-1}{e^{2\dpB\hat{x}^{\alpha}/\dpb_{\gamma}}+1}\right) = x^{\gamma}\coth\left(\dpB\hat{x}^{\alpha}/\dpb_{\gamma}\right) .
    \end{equation*}
    
    In the case where $\alpha=\gamma$ in (\ref{eqn: Bn rat hat cont}), we can use expression (\ref{eqn: Bn rat xinv1}) to substitute for $x^{i}$ with $i\neq\gamma$ as follows:
    \begin{align*}
    \hat{x}^{\gamma} & = \frac{\dpb_{0}+\dpb_{\gamma}}{\dpB}\log(x^{\gamma})  + \sum_{i \neq 0,\gamma} \frac{\dpb_{i}}{2\dpB}\log\big((x^{\gamma})^{2}(1-\coth^{2}\left(\dpB\hat{x}^{i}/\dpb_{\gamma}\right)\big) \\
    & = \frac{1}{\dpB}\left(\dpb_{0}+\dpb_{\gamma}+\sum_{i \neq 0,\gamma}\dpb_{i}\right)\log(x^{\gamma}) + \sum_{i \neq 0,\gamma} \frac{\dpb_{i}}{2\dpB}\log\left((1-\coth^{2}\left(\dpB\hat{x}^{i}/\dpb_{\gamma}\right)\right) .
    \end{align*}
    From this, we obtain (\ref{eqn: Bn rat xinv2}).
\end{proof}

From the definition of a Legendre transform, we have that 
\begin{equation*}
    F_{\alpha\beta} = \widehat{F}_{\alpha\beta} .
\end{equation*}
We will use this, along with Lemma \ref{thm: Bn rat xinv}, to compute all the second-order derivatives of $\hat{F}$ in terms of the $\hat{x}^{\alpha}$.

\begin{lem}\label{thm: Bn rat Fab}
    The second-order derivatives $\widehat{F}_{\alpha\beta}$ can be expressed, up to constant terms, by the following formulae for all $\alpha,\beta \in \{1,\dots,n\}$. \newline
    Case 1: $\alpha=\gamma$.
    \begin{equation*}
        \widehat{F}_{\gamma\beta} = 4\dpb_{\beta}\dpB\hat{x}^{\beta} .
    \end{equation*}
    Case 2: $\alpha=\beta\neq\gamma$.
    \begin{multline*}
        \widehat{F}_{\alpha\alpha} \overset{\text{\tiny$\bullet$}}{=}
        4\dpb_{\alpha}\dpB\hat{x}^{\gamma} + \frac{4\dpb_{\alpha}\dpB}{\dpb_{\gamma}}(\dpB-\dpb_{0}-2\dpb_{\alpha})\hat{x}^{\alpha} \\
        + 4\dpb_{\alpha}(\dpb_{0}+\dpB)\log\left(1-e^{2\dpB\hat{x}^{\alpha}/\dpb_{\gamma}}\right) + 4\dpb_{\alpha}(\dpb_{\alpha}-\dpB)\log\left(1-e^{4\dpB\hat{x}^{\alpha}/\dpb_{\gamma}}\right) \\
        + 2\dpb_{\alpha}\sum_{\substack{i=1\\i \neq \alpha,\gamma}}^{n} \dpb_{i}\left[\log\left(1-e^{2\dpB(\hat{x}^{i}-\hat{x}^{\alpha})/\dpb_{\gamma}}\right)+\log\left(1-e^{2\dpB(\hat{x}^{i}+\hat{x}^{\alpha})/\dpb_{\gamma}}\right) - \frac{2\dpB}{\dpb_{\gamma}}\hat{x}^{i}\right] .
    \end{multline*}
    Case 3: $\alpha, \beta, \gamma$ distinct.
    \begin{equation*}
        \widehat{F}_{\alpha\beta} \overset{\text{\tiny$\bullet$}}{=} 2\dpb_{\alpha}\dpb_{\beta}\left[\log\left(1-e^{2\dpB(\hat{x}^{\alpha}+\hat{x}^{\beta})/\dpb_{\gamma}}\right) - \log\left(1-e^{2\dpB(\hat{x}^{\alpha}-\hat{x}^{\beta})/\dpb_{\gamma}}\right) - \frac{2\dpB}{\dpb_{\gamma}}\hat{x}^{\beta}\right] .
    \end{equation*}
\end{lem}
\begin{proof}
    Case 1: $\alpha=\gamma$. \newline
    By definition of the new coordinate system, we have
    \begin{equation*}
        \hat{F}_{\gamma\beta} \overset{\text{\tiny$\bullet$}}{=} \hat{x}_{\beta} =\eta_{\beta\alpha}\hat{x}^{\alpha} = \eta_{\beta\beta}\hat{x}^{\beta},
    \end{equation*}
    since the metric $\eta$ is diagonal.
    
    Case 2: $\alpha=\beta$ and $\alpha\neq\gamma$. \newline
    Calculating $\widehat{F}_{\alpha\alpha}$ directly, we have
    \begin{equation*}
        \widehat{F}_{\alpha\alpha} = F_{\alpha\alpha} \overset{\text{\tiny$\bullet$}}{=} 4\dpb_{\alpha}(\dpb_{\alpha}+\dpb_{0})\log(x^{\alpha})+2\dpb_{\alpha}\sum_{\substack{i=1\\i \neq \alpha}}^{n} \dpb_{i}\log\left((x^{i})^{2}-(x^{\alpha})^{2}\right) .
    \end{equation*}
    We can use expression (\ref{eqn: Bn rat xinv1}) to make substitutions, and collect terms in $\log\left(x^{\gamma}\right)$ as follows:
    \begin{multline*}
        \widehat{F}_{\alpha\alpha} \overset{\text{\tiny$\bullet$}}{=} 4\dpb_{\alpha}\dpB\log(x^{\gamma}) + 4\dpb_{\alpha}(\dpb_{\alpha}+\dpb_{0})\log\left(\coth\left(\frac{\dpB}{\dpb_{\gamma}}\hat{x}^{\alpha}\right)\right) \\
         + 2\dpb_{\gamma}\dpb_{\alpha}\log\left(1-\coth^{2}\left(\frac{\dpB}{\dpb_{\gamma}}\hat{x}^{\alpha}\right)\right) \\
         + 2\dpb_{\alpha}\sum_{\substack{i=1\\i \neq \alpha,\gamma}}^{n} \dpb_{i}\log\left(\coth^{2}\left(\frac{\dpB}{\dpb_{\gamma}}\hat{x}^{i}\right)-\coth^{2}\left(\frac{\dpB}{\dpb_{\gamma}}\hat{x}^{\alpha}\right)\right).
    \end{multline*}
    Next, we use (\ref{eqn: Bn rat xinv2}) to substitute for the terms in $\log\left(x^{\gamma}\right)$:
    \begin{multline*}
        \widehat{F}_{\alpha\alpha} \overset{\text{\tiny$\bullet$}}{=} 4\dpb_{\alpha}\dpB\hat{x}^{\gamma} + 4\dpb_{\alpha}(\dpb_{\alpha}+\dpb_{0})\log\left(\coth\left(\frac{\dpB}{\dpb_{\gamma}}\hat{x}^{\alpha}\right)\right) \\
        + 2\dpb_{\alpha}(\dpb_{\gamma}-\dpb_{\alpha})\log\left(1-\coth^{2}\left(\frac{\dpB}{\dpb_{\gamma}}\hat{x}^{\alpha}\right)\right) \\
        + 2\dpb_{\alpha}\sum_{\substack{i=1\\i \neq \alpha,\gamma}}^{n} \dpb_{i}\log\left(\frac{\coth^{2}\left(\dpB\hat{x}^{i}/\dpb_{\gamma}\right)-\coth^{2}\left(\dpB\hat{x}^{\alpha}/\dpb_{\gamma}\right)}{1-\coth^{2}\left(\dpB\hat{x}^{i}/\dpb_{\gamma}\right)}\right) .
    \end{multline*}
    Since $\coth(u) = \frac{e^{2u}-1}{e^{2u}+1}$, we note that
    \begin{align*}
        1-\coth^{2}(u) &= \frac{4e^{2u}}{(e^{2u}+1)^{2}} \\
    \intertext{and}
        \coth^{2}(u)-\coth^{2}(v) &= \frac{4e^{2v}(e^{2(u-v)}-1)(e^{2(u+v)}-1)}{(e^{2u}+1)^{2}(e^{2v}+1)^{2}} .
    \end{align*}
    This allows us to expand and simplify $\widehat{F}_{\alpha\alpha}$ as follows:
    \begin{multline*}
        \widehat{F}_{\alpha\alpha} \overset{\text{\tiny$\bullet$}}{=} 4\dpb_{\alpha}\dpB\hat{x}^{\gamma} + \frac{4\dpb_{\alpha}\dpB}{\dpb_{\gamma}}(\dpB-\dpb_{0}-2\dpb_{\alpha})\hat{x}^{\alpha} - \frac{4\dpb_{\alpha}\dpB}{\dpb_{\gamma}}\sum_{\substack{i=1\\i \neq \alpha,\gamma}}^{n}\dpb_{i}\hat{x}^{i} \\
        + 4\dpb_{\alpha}(\dpb_{\alpha}+\dpb_{0})\log\left(e^{2\dpB\hat{x}^{\alpha}/\dpb_{\gamma}}-1\right) + 4\dpb_{\alpha}(\dpb_{\alpha}-\dpB)\log\left(e^{2\dpB\hat{x}^{\alpha}/\dpb_{\gamma}}+1\right) \\
        + 2\dpb_{\alpha}\sum_{\substack{i=1\\i \neq \alpha,\gamma}}^{n} \dpb_{i}\left[\log\left(e^{2\dpB(\hat{x}^{i}-\hat{x}^{\alpha})/\dpb_{\gamma}}-1\right)+\log\left(e^{2\dpB(\hat{x}^{i}+\hat{x}^{\alpha})/\dpb_{\gamma}}-1\right)\right] .
    \end{multline*}
    We may rewrite this expression using the following identity:
    \begin{equation*}
        \log(1+e^{u}) = \log(1-e^{2u})-\log(1-e^{u}).
    \end{equation*}
    This leads to the required statement.
    
    Case 3: $\alpha, \beta, \gamma$ distinct. \newline
    By direct calculation, we have
    \begin{equation*}
        \widehat{F}_{\alpha\beta} = F_{\alpha\beta} \overset{\text{\tiny$\bullet$}}{=} 2\dpb_{\alpha}\dpb_{\beta}\log\left(x^{\alpha}+x^{\beta}\right) - 2\dpb_{\alpha}\dpb_{\beta}\log\left(x^{\alpha}-x^{\beta}\right) ,
    \end{equation*}
    and we can make substitutions using Lemma \ref{thm: Bn rat xinv} to find
    \begin{multline}\label{eqn.Bnrat.Fabproof}
        \widehat{F}_{\alpha\beta} \overset{\text{\tiny$\bullet$}}{=} 2\dpb_{\alpha}\dpb_{\beta}\log\left(\coth\left(\frac{\dpB}{\dpb_{\gamma}}\hat{x}^{\alpha}\right)+\coth\left(\frac{\dpB}{\dpb_{\gamma}}\hat{x}^{\beta}\right)\right) \\
         - 2\dpb_{\alpha}\dpb_{\beta}\log\left(\coth\left(\frac{\dpB}{\dpb_{\gamma}}\hat{x}^{\alpha}\right) - \coth\left(\frac{\dpB}{\dpb_{\gamma}}\hat{x}^{\beta}\right)\right) .
    \end{multline}
    We note the following identities:
    \begin{align*}
        \coth(u)+\coth(v) & = \frac{2(e^{2(u+v)}-1)}{(e^{2u}+1)(e^{2v}+1)} \\
        \intertext{and}
        \coth(u)-\coth(v) & = \frac{2e^{2v}(e^{2(u-v)}-1)}{(e^{2u}+1)(e^{2v}+1)} .
    \end{align*}
    Using these, we can reformulate (\ref{eqn.Bnrat.Fabproof}) to obtain the required expression.
\end{proof}

\begin{thm}
    The Legendre transform $S_{\gamma}$, for arbitrary $\gamma\in\{1,\dots,n\}$, applied to the function $F$ as defined in (\ref{eqn: Bn rat sol}) has the form
    \begin{equation}\label{eqn: Bn rat LT}
    \begin{aligned}
        \widehat{F} & = \frac{\eta_{\gamma\gamma}}{6}\left(\hat{x}^{\gamma}\right)^{3} + \sum_{\substack{i=1 \\ i\neq\gamma}}^{n}\frac{\eta_{ii}}{2}\hat{x}^{\gamma}\left(\hat{x}^{i}\right)^{2} + \sum_{\substack{i<j \\ i,j\neq \gamma}}^{n}\frac{2\dpb_{i}\dpb_{j}\dpb_{\gamma}^{2}}{\dpB^{2}}f\left(-\frac{\dpB}{\dpb_{\gamma}}\left(\hat{x}^{i}\pm\hat{x}^{j}\right)\right)\\
        & + \sum_{\substack{i=1 \\ i\neq\gamma}}^{n}\left[\frac{4\dpb_{i}\dpb_{\gamma}^{2}(\dpb_{0}+\dpB)}{\dpB^{2}}f\left(-\frac{\dpB}{\dpb_{\gamma}}\hat{x}^{i}\right) + \frac{\dpb_{i}\dpb_{\gamma}^{2}\left(\dpb_{i}-\dpB\right)}{\dpB^{2}}f\left(-\frac{2\dpB}{\dpb_{\gamma}}\hat{x}^{i}\right)\right] ,
    \end{aligned}
    \end{equation}
    up to quadratic terms and coordinate shifts. The function $f(z)$ is as defined in (\ref{eqn: f def}).
\end{thm}
\begin{proof}
    Taking (\ref{eqn: Bn rat LT}) as an ansatz, we can compare its second-order derivatives with the expressions given in Lemma \ref{thm: Bn rat Fab}. Note that we use identity (\ref{eqn.Anrat.f''}) to express logarithmic terms as the function $f(z)$.
\end{proof}
\noindent Comparison of this solution with other trigonometric solutions is done in Section \ref{sn.Bnr2t}.

\section{Multi-parameter \texorpdfstring{$A_{n}$}{An}-type trigonometric solutions}\label{sn.Antrig}
In this section, we generalise solutions found by Hoevenaars and Martini; see Theorem 2.1 in the preprint version of \cite{HM03-5D} for full details. These solutions were also considered by Shen \cite{S19-TRS}.

We introduce a multi-parameter deformation of these solutions. The corresponding function $F=F^{\textrm{trig}}_{\hmm}$ depends on the $n$-tuple $\hmm=(\hmm_{1},\dots,\hmm_{n})$, $\hmm_{i}\in \CC^{\times} \,\forall\, i$, as well as on three parameters $\ahm,\bhm,\chm \in \CC$. The non-polynomial part of $F^{\textrm{trig}}_{\hmm}$ is expressed via the function $f(z)$ given by formula (\ref{eqn: f def}).
Relations between $\ahm,\bhm,\chm$ and $M=\sum\hmm_{i}$ are needed to ensure that $F^{\textrm{trig}}_{\hmm}$ solves the WDVV equations. We will first consider a generic case before looking at other situations.
\begin{thm}\label{thm: An trig}
    Suppose that 
    \begin{align}
        \bhm\hmM+\chm & \neq 0 , \label{eqn: An trig cond1}\\
        \intertext{and}
        \ahm\hmM^{2}+3\bhm\hmM+\chm & \neq 0 .\label{eqn: An trig cond2}
    \end{align}
    Then $\eta = \sum_{k = 1}^{n} F_{k}$ is a constant non-degenerate matrix.
    Furthermore, the function
    \begin{multline}\label{eqn: An trig sol}
        F^{\emph{trig}}_{\hmm}(\hmc) = \sum_{1\leq i < j\leq n}\hmm_{i}\hmm_{j}f\left(\hmc_{i}-\hmc_{j}\right) + \frac{\ahm}{6}\left(\sum_{i=1}^{n}\hmm_{i}\hmc_{i}\right)^{3} \\
        + \frac{\bhm}{2}\left(\sum_{i=1}^{n}\hmm_{i}\hmc_{i}\right)\left(\sum_{j=1}^{n}\hmm_{j}\hmc_{j}^{2}\right) + \frac{\chm}{6}\sum_{i=1}^{n}\hmm_{i}\hmc_{i}^{3} 
    \end{multline}
    solves the WDVV equations (\ref{eqn.intro.WDVV}) if the following relation holds:
    \begin{equation}\label{eqn: An trig condWDVV}
        \bhm^{3}\hmM + 3\bhm^{2}\chm - \ahm\chm^{2} + \ahm\hmM^{2} + 3\bhm\hmM + \chm = 0 .
    \end{equation}
    Conversely, the WDVV equations (\ref{eqn.intro.WDVV}) imply relation (\ref{eqn: An trig condWDVV}) when $n \geq 3$.
\end{thm}
\noindent In the case when $m_{i} = 1$ for all $i$, Theorem \ref{thm: An trig} is equivalent to Theorem 2.1 in Hoevenaars and Martini's work. 

From the definition of $f(z)$ in (\ref{eqn: f def}) we have that $f'''(z)=\coth{z}$, and so
\begin{equation}\label{eqn: An trig f minus}
    f'''(-z) = -f'''(z).
\end{equation}
For convenience, we will use the following shorthand throughout the calculations:
\begin{equation}\label{eqn: An trig beta def}
    \beta_{ij} =
    \begin{cases}
        \coth{(\hmc_{i}-\hmc_{j})} & \textrm{if } i\neq j , \\ 
        0 & \textrm{otherwise} .
    \end{cases}
\end{equation}
Let $\delta_{ij}$ be the Kronecker delta function.
We also define 
\begin{equation*}
    \delta_{ijr} = \delta_{ij}\delta_{jr} = 
    \begin{cases}
        1 & \textrm{if } i=j=r , \\
        0 & \textrm{otherwise.}
    \end{cases}
\end{equation*}

First, we compute the matrices of third order derivatives of $F$. As previously, let $F_{k}$ be the $n\times n$ matrix with $\left(r,s\right)^{\text{th}}$ entry
\begin{equation*}
    F_{krs} = \frac{\partial^{3}F}{\partial \hmc_{k} \partial\hmc_{r} \partial\hmc_{s}}.
\end{equation*}

\begin{lem}\label{thm: An trig Fk}
    The matrix $F_{k}$ of third order derivatives $F_{krs}$ of the function $F=F^{\emph{trig}}_{\hmm}$ given by (\ref{eqn: An trig sol}) can be written as
    \begin{equation*}
        F_{k} = W_{k} + V_{k},
    \end{equation*}  
    where the $n\times n$ matrices $W_{k}$ and $V_{k}$ have $\left(r,s\right)^{th}$ entry $W_{krs}$ and $V_{krs}$ respectively, given by
    \begin{align*}
        W_{krs} & = \ahm\hmm_{k}\hmm_{r}\hmm_{s}, \\
        V_{krs} & = \delta_{krs}\hmm_{k}\left(\sum_{q=1}^{n}\hmm_{q}\beta_{kq} +\chm\right) + \delta_{kr}\hmm_{k}\hmm_{s}\beta_{sk} + \delta_{ks}\hmm_{k}\hmm_{r}\beta_{rk} + \delta_{rs}\hmm_{k}\hmm_{r}\beta_{kr} \\
        & \hspace{14em} + \delta_{kr}\bhm\hmm_{k}\hmm_{s} + \delta_{ks}\bhm\hmm_{k}\hmm_{r} + \delta_{rs}\bhm\hmm_{k}\hmm_{r}.
    \end{align*}
\end{lem}
\begin{proof}
    Considering each term in (\ref{eqn: An trig sol}) individually, we can set out the following collection of statements:
    \begin{align}
        & \partial_{k}\partial_{r}\partial_{s}\left(\sum_{i<j}\hmm_{i}\hmm_{j}f(\hmc_{i}-\hmc_{j}) \right) = \delta_{krs}\sum_{q=1}^{n}\hmm_{k}\hmm_{q}\beta_{kq} + \delta_{kr}(1-\delta_{rs})\hmm_{k}\hmm_{s}\beta_{sk} \nonumber \\
        & \hspace{9em} + \delta_{ks}(1-\delta_{rs})\hmm_{k}\hmm_{r}\beta_{rk} + \delta_{rs}(1-\delta_{kr})\hmm_{k}\hmm_{r}\beta_{kr}, \label{eqn.Antrig.Vsum1} \\
        & \partial_{k}\partial_{r}\partial_{s}\left(\frac{\ahm}{6}\left(\sum_{i=1}^{n}\hmm_{i}\hmc_{i}\right)^{3}\right) = \ahm \hmm_{k}\hmm_{r}\hmm_{s} , \label{eqn.Antrig.Wsum} \\
        & \partial_{k}\partial_{r}\partial_{s}\left(\frac{\bhm}{2}\left(\sum_{i=1}^{n}\hmm_{i}\hmc_{i}\right)\left(\sum_{j=1}^{n}\hmm_{j}\hmc_{j}^{2}\right)\right) =
         3\bhm\delta_{krs}\hmm_{k}^{2} + \delta_{kr}(1-\delta_{rs})\bhm\hmm_{k}\hmm_{s} \nonumber \\
         & \hspace{11em} + \delta_{ks}(1-\delta_{rs})\bhm\hmm_{k}\hmm_{r} + \delta_{rs}(1-\delta_{kr})\bhm\hmm_{r}\hmm_{k} , \label{eqn.Antrig.Vsum2} \\
    \intertext{and}
        & \partial_{k}\partial_{r}\partial_{s}\left(\frac{\chm}{6}\sum_{i=1}^{n}\hmm_{i}\hmc_{i}^{3}\right) = \delta_{krs}\chm . \label{eqn.Antrig.Vsum3}
    \end{align}
    The expression (\ref{eqn.Antrig.Wsum}) is equal to $W_{krs}$, and the sum of (\ref{eqn.Antrig.Vsum1}), (\ref{eqn.Antrig.Vsum2}) and (\ref{eqn.Antrig.Vsum3}) equals $V_{krs}$. Note that we can use the fact that $\delta_{ab}\beta_{ab}=0$ to simplify (\ref{eqn.Antrig.Vsum1}) as required.
\end{proof}
For the WDVV equations to hold, we now need to check that the linear combination $\eta = \sum_{k=1}^{n}F_{k}$ is a constant, non-degenerate matrix.

\begin{lem}\label{thm: An trig B def}
    The $n \times n$ matrix $\eta=\sum_{k=1}^{n}F_{k}$ has $\left(r,s\right)^{\text{th}}$ entry
    \begin{equation*}
        \eta_{rs} = (\ahm\hmM+2\bhm)\hmm_{r}\hmm_{s} + \delta_{rs}(\bhm\hmM+\chm)\hmm_{r}.
    \end{equation*}
    In particular, the $\left(r,s\right)^{\text{th}}$ entry of the matrix $\sum_{k=1}^{n}V_{k}$ is given by
    \begin{equation}\label{eqn.Antrig.sumV}
        \sum_{k=1}^{n}V_{krs} = 2\bhm\hmm_{r}\hmm_{s} + \delta_{rs}(\bhm\hmM+\chm)\hmm_{r}.
    \end{equation}
\end{lem}
\begin{proof}
    By Lemma \ref{thm: An trig Fk}, we have $\eta=\sum_{k}W_{k}+\sum_{k}V_{k}$.
    It is clear that
    \begin{equation*}
        \sum_{k=1}^{n}W_{krs} = \ahm\hmM\hmm_{r}\hmm_{s},
    \end{equation*}
    Additionally, we have that
    \begin{multline*}
        \sum_{k=1}^{n}V_{krs} = \delta_{rs}\hmm_{r}\left(\sum_{q=1}^{n}\hmm_{q}\beta_{rq} +\chm\right) + \hmm_{r}\hmm_{s}\beta_{sr} + \hmm_{s}\hmm_{r}\beta_{rs} \\
        + \delta_{rs}\hmm_{r}\sum_{k=1}^{n}\hmm_{k}\beta_{kr} + 2\bhm\hmm_{r}\hmm_{s} + \delta_{rs}\bhm\hmM\hmm_{r},
    \end{multline*}
    which can be simplified to the form (\ref{eqn.Antrig.sumV}).
\end{proof}

Note that the matrix $\eta$ can be represented as 
\begin{equation*}
    \eta = A + uv^{\text{T}},
\end{equation*}
where the $n\times n$ matrix $A$ has the entries
\begin{equation}\label{eqn: An trig A def}
    A_{rs} = \delta_{rs}(\bhm\hmM+\chm)\hmm_{r}
\end{equation}
and $u=(u_{1},\dots,u_{n})^{T}$, $v=(v_{1},\dots,v_{n})^{T}$ are $n$-dimensional column vectors such that
\begin{equation*}
    u_{i} = \hmm_{i} \textrm{ and } v_{i} = (\ahm\hmM+2\bhm)\hmm_{i} .
\end{equation*}

\begin{lem}\label{thm: An trig B-1}
    Suppose the matrix $\eta$ from Lemma \ref{thm: An trig B def} is non-singular. Then 
    \begin{equation*}
        \det \eta = \left(\ahm\hmM^{2}+3\bhm\hmM+\chm\right)\left(\bhm\hmM+\chm\right)^{n-1}\prod_{i=1}^{n}\hmm_{i} ,
    \end{equation*}
    and the inverse matrix $\eta^{-1}$ has entries
    \begin{equation}\label{eqn: An trig B-1}
        \eta^{rs} = \frac{\delta_{rs}}{(\bhm\hmM+\chm)\hmm_{r}} - \frac{\ahm\hmM+2\bhm}{(\ahm\hmM^{2}+3\bhm\hmM+\chm)(\bhm\hmM+\chm)} .
    \end{equation}
\end{lem}
\begin{proof}
    Since we can write $\eta = A+uv^{\text{T}}$, this implies that
    \begin{equation}\label{eqn: An trig detB}
        \det \eta  = \det (A+uv^{T}).
    \end{equation}
    By applying the Matrix Determinant Lemma (see Thm. 18.1.1 in \cite{H08-MA}), we have
    \begin{equation*}
        \det \eta = (1 + v^{T} A^{-1} u) \det A .
    \end{equation*}
    The matrix $A$, given in (\ref{eqn: An trig A def}), has determinant
    \begin{equation}\label{eqn: An trig detA}
        \det A = \left(\bhm\hmM+\chm\right)^{n}\prod_{i=1}^{n}\hmm_{i} .
    \end{equation}
    As $A$ is a diagonal matrix, it is straightforward to see that $A^{-1}$ exists and has the entries
    \begin{equation}\label{eqn: An trig A-1}
        (A^{-1})_{rs} = \frac{\delta_{rs}}{(\bhm\hmM+\chm)\hmm_{r}} .
    \end{equation}
    We can now compute
    \begin{equation}\label{eqn: An trig SM1}
        1 + v^{T} A^{-1} u = 1 + \sum_{i=1}^{n}\sum_{j=1}^{n} \frac{(\ahm\hmM+2\bhm)\hmm_{i}\delta_{ij}}{(\bhm\hmM+\chm)} = \frac{\ahm\hmM^{2}+3\bhm\hmM+\chm}{\bhm\hmM+\chm} .
    \end{equation}
    It follows from (\ref{eqn: An trig detB}), (\ref{eqn: An trig detA}) and (\ref{eqn: An trig SM1}) that
    \begin{equation*}
        \det \eta = \left(\ahm\hmM^{2}+3\bhm\hmM+\chm\right)\left(\bhm\hmM+\chm\right)^{n-1}\prod_{i=1}^{n}\hmm_{i} .
    \end{equation*}
    
    We have $1+v^{T}A^{-1}u\neq 0$, since this is equivalent to the condition that $\ahm\hmM^{2}+3\bhm\hmM+\chm \neq 0$, which is satisfied as $\eta$ is invertible. It is then easy to check (cf. the Sherman-Morrison formula, see e.g. \cite{H89-SM}) that
    \begin{equation}\label{eqn: An trig SM}
        \eta^{-1} = A^{-1} - \frac{A^{-1}uv^{T}A^{-1}}{1+v^{T}A^{-1}u} .
    \end{equation}
    We have
    \begin{equation}\label{eqn: An trig SM2}
    \begin{aligned}
        \left(A^{-1}uv^{T}A^{-1}\right)_{rs} & = \frac{\ahm\hmM+2\bhm}{(\bhm\hmM+\chm)^{2}}\sum_{i=1}^{n}\sum_{j=1}^{n}\frac{\delta_{ri}\delta_{js}}{\hmm_{r}\hmm_{s}}\hmm_{i}\hmm_{j} \\
        & = \frac{\ahm\hmM+2\bhm}{(\bhm\hmM+\chm)^{2}} ,
    \end{aligned}
    \end{equation}
    and then formulas (\ref{eqn: An trig A-1})--(\ref{eqn: An trig SM2}) imply (\ref{eqn: An trig B-1}).
\end{proof}

For the WDVV equations (\ref{eqn.intro.WDVV}) to hold we need
\begin{equation}\label{eqn: An trig WDVV eq}
    F_{i}\eta^{-1}F_{j} - F_{j}\eta^{-1}F_{i} = 0.
\end{equation}
From Lemma \ref{thm: An trig B-1}, we have for non-singular $\eta$
\begin{align}
    \left(F_{i}\eta^{-1}F_{j}\right)_{rs} & = \sum_{k=1}^{n}\sum_{l=1}^{n}\left[F_{irk}\left(\kappa + \frac{\delta_{kl}}{(\bhm\hmM+\chm)\hmm_{k}}\right)F_{jls}\right] \nonumber\\ 
    & = \kappa\sum_{k=1}^{n}F_{irk}\sum_{l=1}^{n}F_{jls} + \frac{1}{\bhm\hmM+\chm}\sum_{k=1}^{n}\frac{1}{\hmm_{k}}F_{irk}F_{jks} \nonumber\\
    & = \kappa \, \eta_{ri} \eta_{js} + \frac{1}{\bhm\hmM+\chm}\sum_{k=1}^{n}\frac{1}{\hmm_{k}}\Big[W_{irk}W_{jks} + W_{irk}V_{jks} \label{eqn: An trig FBF}\\
    & \hspace{17em}+ V_{irk}W_{jks} + V_{irk}V_{jks}\Big] , \nonumber
\end{align}
where
\begin{equation*}
    \kappa = - \frac{\ahm\hmM+2\bhm}{(\ahm\hmM^{2}+3\bhm\hmM+\chm)(\bhm\hmM+\chm)} .
\end{equation*}
We will consider contributions to (\ref{eqn: An trig WDVV eq}) from each of the five terms in (\ref{eqn: An trig FBF}) individually.

For a four-index expression $M_{ijkl}$, we denote by $M_{[ij]kl}$ the anti-symmetrisation \newline $M_{[ij]kl}=M_{ijkl}-M_{jikl}$.

\begin{lem}\label{thm: An trig FBF1}
We have
\begin{multline*}
    \left(F_{i}\eta^{-1}F_{j}-F_{j}\eta^{-1}F_{i}\right)_{rs} = \lambda_{[ij]rs} + \kappa(\bhm\hmM+\chm)^{2}\hmm_{r}\hmm_{s}(\delta_{ir}\delta_{js}-\delta_{jr}\delta_{is}) \\ + \big(\kappa(\ahm\hmM+2\bhm)(\bhm\hmM+\chm)+\ahm\big)\hmm_{r}\hmm_{s}\big(\hmm_{j}(\delta_{ir}-\delta_{is}) +\hmm_{i}(\delta_{js}-\delta_{jr})\big) ,
\end{multline*}
where
\begin{equation}\label{eqn.Antrig.VVdef}
    \lambda_{ijrs} = \frac{1}{\bhm\hmM+\chm}\sum_{k=1}^{n}\frac{1}{\hmm_{k}}\left(V_{i}\right)_{rk}\left(V_{j}\right)_{ks} .
\end{equation}
\end{lem}
\begin{proof}
    We start with the first term in (\ref{eqn: An trig FBF}). From Lemma \ref{thm: An trig B def} we have
    \begin{multline*}
        \eta_{ri}\eta_{js}  = (\ahm\hmM+2\bhm)^{2}\hmm_{r}\hmm_{s}\hmm_{i}\hmm_{j} + \delta_{ir}(\bhm\hmM+\chm)(\ahm\hmM+2\bhm)\hmm_{r}\hmm_{j}\hmm_{s} \\
         + \delta_{js}(\bhm\hmM+\chm)(\ahm\hmM+2\bhm)\hmm_{s}\hmm_{r}\hmm_{i} + \delta_{ir}\delta_{js}(\bhm\hmM+\chm)^{2}\hmm_{r}\hmm_{s},
    \end{multline*}
    so the antisymmetrisation becomes
    \begin{multline}\label{eqn: An trig BB}
        \eta_{ri} \eta_{js} - \eta_{rj} \eta_{is} = (\bhm\hmM+\chm)^{2}\hmm_{r}\hmm_{s}(\delta_{ir}\delta_{js}-\delta_{jr}\delta_{is})\\
        + (\ahm\hmM+2\bhm)(\bhm\hmM+\chm)\hmm_{r}\hmm_{s}\big(\hmm_{j}(\delta_{ir}-\delta_{is}) +\hmm_{i}(\delta_{js}-\delta_{jr})\big) .
    \end{multline}
    Next, we use Lemma \ref{thm: An trig Fk} to consider the other terms in (\ref{eqn: An trig FBF}). We have
    \begin{equation*}
        \sum_{k=1}^{n}\frac{1}{\hmm_{k}}W_{irk}W_{jks} = \ahm^{2}\hmm_{i}\hmm_{r}\hmm_{j}\hmm_{s}\hmM , 
    \end{equation*}
    which leads to
    \begin{equation}\label{eqn: An trig WW}
        \sum_{k=1}^{n}\frac{1}{\hmm_{k}}\big[W_{irk}W_{jks} - W_{jrk}W_{iks}\big] = 0.
    \end{equation}
    We continue with
    \begin{equation*}
        \sum_{k=1}^{n}\frac{1}{\hmm_{k}}W_{irk}V_{jks} = \ahm\hmm_{i}\hmm_{r}\sum_{k=1}^{n}V_{kjs} ,
    \end{equation*}
    since it follows from Lemma \ref{thm: An trig Fk} that $V_{kjs}$ is symmetric in $k,j,$ and $s$. By Lemma \ref{thm: An trig B def},
    \begin{align*}
        \sum_{k=1}^{n}\frac{1}{\hmm_{k}}W_{irk}V_{jks} & = \ahm\hmm_{i}\hmm_{r}\big(2\bhm\hmm_{j}\hmm_{s} + \delta_{js}(\bhm\hmM+\chm)\hmm_{j}\big) \\
        & = 2\ahm\bhm\hmm_{i}\hmm_{j}\hmm_{r}\hmm_{s} + \delta_{js}\ahm(\bhm\hmM+\chm)\hmm_{i}\hmm_{j}\hmm_{r} ,
    \end{align*}
    and so we find
    \begin{equation}\label{eqn: An trig WV}
        \sum_{k=1}^{n}\frac{1}{\hmm_{k}}\big[W_{irk}V_{jks} - W_{jrk}V_{iks}\big] = \ahm(\bhm\hmM+\chm)\hmm_{r}\hmm_{s}\left(\delta_{js}\hmm_{i} - \delta_{is}\hmm_{j}\right) .
    \end{equation}
    By the symmetry of $V_{ijk}$ and $W_{ijk}$, we similarly obtain
    \begin{equation}\label{eqn: An trig VW}
        \sum_{k=1}^{n}\frac{1}{\hmm_{k}}\big[V_{irk}W_{jks} - V_{jrk}W_{iks}\big] = \ahm(\bhm\hmM+\chm)\hmm_{r}\hmm_{s}(\delta_{ir}\hmm_{j}-\delta_{jr}\hmm_{i}).
    \end{equation}
    The antisymmetrisation of (\ref{eqn: An trig FBF}) is then found by combining (\ref{eqn: An trig BB}), (\ref{eqn: An trig WW}), (\ref{eqn: An trig WV}), and (\ref{eqn: An trig VW}) with $\lambda_{[ij]rs}$ to represent the remaining terms.
\end{proof}

To find the antisymmetrisation of $\lambda_{ijrs}$, given in (\ref{eqn.Antrig.VVdef}), we now separately consider terms in the expansion of this expression which are (without the use of any identities) quadratic, linear and constant in $\beta_{pq}$. It will turn out by the use of various identities that all terms in $\lambda_{[ij]rs}$ are constant.

\begin{lem}\label{thm: An trig VV}
    We have
    \begin{multline*}
        \lambda_{[ij]rs} = \frac{\hmm_{r}\hmm_{s}}{\bhm\hmM+\chm}\bigg(\left(\hmM\left(\bhm^{2}+1\right)+2\bhm\chm\right)(\delta_{ir}\delta_{js}-\delta_{jr}\delta_{is}) \\
        + \left(1-\bhm^{2}\right)\hmm_{i}(\delta_{jr}-\delta_{js}) + \left(1-\bhm^{2}\right)\hmm_{j}(\delta_{is}-\delta_{ir})
        \bigg) .
    \end{multline*}
\end{lem}
\begin{proof}
    Let $A_{[ij]rs}$ be the terms in $\lambda_{[ij]rs}$ which are quadratic in $\beta_{pq}$. That is,
    \begin{multline*}
        A_{ijrs} = \sum_{k=1}^{n}\hmm_{k}\left(\delta_{kir}\sum_{q=1}^{n}\hmm_{q}\beta_{kq} + \delta_{ki}\hmm_{r}\beta_{rk} + \delta_{kr}\hmm_{i}\beta_{ik} + \delta_{ir}\hmm_{i}\beta_{ki}\right) \\
        \times\left(\delta_{kjs}\sum_{q=1}^{n}\hmm_{q}\beta_{kq} + \delta_{kj}\hmm_{s}\beta_{sk} + \delta_{ks}\hmm_{j}\beta_{jk} + \delta_{js}\hmm_{j}\beta_{kj}\right),
    \end{multline*}
    where we have used the expression for $\lambda_{ijrs}$ given in (\ref{eqn.Antrig.VVdef}), and for $V_{q}$ given in Lemma \ref{thm: An trig Fk}.
    Terms in $A_{ijrs}$ with a factor of $\delta_{ij}$ will cancel after antisymmetrisation in $i,j$, so we find
   \begin{equation*}
       A_{[ij]rs} = \Tilde{A}_{[ij]rs},
   \end{equation*}
   where
   \begin{align*}
       \Tilde{A}_{ijrs} &= \delta_{irs}\hmm_{i}\hmm_{j}\beta_{ji}\left(\sum_{q=1}^{n}\hmm_{q}\beta_{iq}\right) + \delta_{ir}\delta_{js}\hmm_{i}\hmm_{j}\beta_{ij}\left(\sum_{q=1}^{n}\hmm_{q}\beta_{iq}\right) \\
       & \;\;\; + \delta_{is}\hmm_{i}\hmm_{j}\hmm_{r}\beta_{ri}\beta_{ji} + \delta_{js}\hmm_{i}\hmm_{j}\hmm_{r}\beta_{ri}\beta_{ij} + \delta_{jrs}\hmm_{i}\hmm_{j}\beta_{ij}\left(\sum_{q=1}^{n}\hmm_{q}\beta_{jq}\right) \\
       & \;\;\; + \delta_{jr}\hmm_{i}\hmm_{j}\hmm_{s}\beta_{ij}\beta_{sj} + \delta_{rs}\hmm_{i}\hmm_{j}\hmm_{r}\beta_{ir}\beta_{jr} + \delta_{js}\hmm_{i}\hmm_{j}\hmm_{r}\beta_{ir}\beta_{rj} \\
       & \;\;\; + \delta_{ir}\delta_{js}\hmm_{i}\hmm_{j}\beta_{ji}\left(\sum_{q=1}^{n}\hmm_{q}\beta_{jq}\right) + \delta_{ir}\hmm_{i}\hmm_{j}\hmm_{s}\beta_{ji}\beta_{sj} \\
       & \;\;\; + \delta_{ir}\hmm_{i}\hmm_{j}\hmm_{s}\beta_{si}\beta_{js} + \delta_{ir}\delta_{js}\hmm_{i}\hmm_{j}\left(\sum_{k=1}^{n}\hmm_{k}\beta_{ki}\beta_{kj}\right) .
   \end{align*}
   Note that the antisymmetrisation of
   \begin{equation*}
       \delta_{irs}\hmm_{i}\hmm_{j}\beta_{ji}\left(\sum_{q=1}^{n}\hmm_{q}\beta_{iq}\right) + \delta_{jrs}\hmm_{i}\hmm_{j}\beta_{ij}\left(\sum_{q=1}^{n}\hmm_{q}\beta_{jq}\right) + \delta_{rs}\hmm_{i}\hmm_{j}\hmm_{r}\beta_{ir}\beta_{jr}
   \end{equation*}
   is equal to zero. Hence,
    \begin{equation}\label{eqn: An trig VVq1}
    A_{[ij]rs}=\Tilde{A}_{[ij]rs}=\Tilde{A}_{[ij]rs}^{(1)}+\Tilde{A}_{[ij]rs}^{(2)}+\Tilde{A}_{[ij]rs}^{(3)} ,
    \end{equation}
    where
    \begin{align*}
        \Tilde{A}_{ijrs}^{(1)} & = \delta_{ir}\delta_{js}\hmm_{i}\hmm_{j}\sum_{q=1}^{n}\hmm_{q}\left(\beta_{ij}\beta_{iq}+\beta_{ij}\beta_{qj}+\beta_{iq}\beta_{jq}\right) , \\
        \Tilde{A}_{ijrs}^{(2)} & = \delta_{jr}\hmm_{i}\hmm_{j}\hmm_{s}\left(\beta_{ij}\beta_{sj}+\beta_{ij}\beta_{is}+\beta_{is}\beta_{js}\right) , \\ 
        \intertext{and} \Tilde{A}_{ijrs}^{(3)} &= \delta_{is}\hmm_{i}\hmm_{j}\hmm_{r}\left(\beta_{ir}\beta_{ij}+\beta_{rj}\beta_{ij}+\beta_{ir}\beta_{jr}\right) .
    \end{align*}
    We now apply the following identity:
    \begin{equation*}
        \beta_{ij}\beta_{ik}+\beta_{ij}\beta_{kj}+\beta_{ik}\beta_{jk}=1,
    \end{equation*}
    where $i,j,k$ are distinct from each other. We get
    \begin{multline*}
        \Tilde{A}_{[ij]rs}^{(1)} = (1-\delta_{ij})(\delta_{ir}\delta_{js}-\delta_{jr}\delta_{is})\hmm_{i}\hmm_{j}\left(\hmM-\hmm_{i}-\hmm_{j}\right) \\
         + (1-\delta_{ij})(\delta_{ir}\delta_{js}-\delta_{jr}\delta_{is})\hmm_{i}\hmm_{j}(\hmm_{i}+\hmm_{j})\beta_{ij}^{2} \\
         + \delta_{ij}(\delta_{ir}\delta_{js}-\delta_{jr}\delta_{is})\hmm_{i}^{2}\left(\sum_{q=1}^{n}\hmm_{q}\beta_{iq}^{2}\right) ,
    \end{multline*}
    and so
    \begin{equation}\label{eqn: An trig VVqA1}
        \Tilde{A}_{[ij]rs}^{(1)} = (\delta_{ir}\delta_{js}-\delta_{jr}\delta_{is})\hmm_{i}\hmm_{j}\left(\hmM-\hmm_{i}-\hmm_{j}+(\hmm_{i}+\hmm_{j})\beta_{ij}^{2}\right)
    \end{equation}
    since $\delta_{ij}(\delta_{ir}\delta_{js}-\delta_{jr}\delta_{is})=0$.
    Furthermore,
    \begin{multline*}
        \Tilde{A}_{ijrs}^{(2)} = \delta_{jr}\delta_{ij}(1-\delta_{is})\hmm_{i}^{2}\hmm_{s}\beta_{is}^{2} + \delta_{jr}\delta_{is}(1-\delta_{ij})\hmm_{i}^{2}\hmm_{j}\beta_{ij}^{2} \\
        + \delta_{jr}\delta_{js}(1-\delta_{ij})\hmm_{i}\hmm_{j}^{2}\beta_{ij}^{2} + \delta_{jr}(1-\delta_{ij})(1-\delta_{is})(1-\delta_{js})\hmm_{i}\hmm_{j}\hmm_{s} ,
    \end{multline*}
    which becomes
    \begin{multline}\label{eqn: An trig VVqA2}
        \Tilde{A}_{ijrs}^{(2)} = \delta_{ijr}\hmm_{i}^{2}\hmm_{s}\beta_{is}^{2} + \delta_{jr}\delta_{is}\hmm_{i}^{2}\hmm_{j}\beta_{ij}^{2} + \delta_{jrs}\hmm_{i}\hmm_{j}^{2}\beta_{ij}^{2} \\
        + \delta_{jr}(1-\delta_{ij})(1-\delta_{is})(1-\delta_{js})\hmm_{i}\hmm_{j}\hmm_{s}.
    \end{multline}
    Similarly,
    \begin{multline}\label{eqn: An trig VVqA3}
        \Tilde{A}_{ijrs}^{(3)} = \delta_{ijs}\hmm_{j}^{2}\hmm_{r}\beta_{jr}^{2} + \delta_{is}\delta_{jr}\hmm_{j}^{2}\hmm_{i}\beta_{ji}^{2} + \delta_{irs}\hmm_{j}\hmm_{i}^{2}\beta_{ji}^{2} \\
         + \delta_{is}(1-\delta_{ij})(1-\delta_{jr})(1-\delta_{ir})\hmm_{i}\hmm_{j}\hmm_{r}.
    \end{multline}
    Note that the first term in the right-hand side of each of (\ref{eqn: An trig VVqA2}) and (\ref{eqn: An trig VVqA3}) becomes zero after antisymmetrisation over $i$ and $j$, and that the second term in each of $\Tilde{A}_{ijrs}^{(2)}$, $-\Tilde{A}_{jirs}^{(2)}$, $\Tilde{A}_{ijrs}^{(3)}$, and $-\Tilde{A}_{jirs}^{(3)}$ cancels with the $\beta_{ij}^{2}$ terms in (\ref{eqn: An trig VVqA1}). Moreover, the third term in the right-hand side of (\ref{eqn: An trig VVqA2}) cancels with the third term in (\ref{eqn: An trig VVqA3}) after antisymmetrisation. Therefore
    \begin{equation}\label{eqn: An trig VVq2}
        \Tilde{A}_{[ij]rs}^{(1)}+\Tilde{A}_{[ij]rs}^{(2)}+\Tilde{A}_{[ij]rs}^{(3)} = (\delta_{ir}\delta_{js}-\delta_{jr}\delta_{is})\hmm_{i}\hmm_{j}\left(\hmM-\hmm_{i}-\hmm_{j}\right) + \alpha_{[ij]rs} ,
    \end{equation}
    where
    \begin{multline*}
        \alpha_{ijrs} = \delta_{jr}(1-\delta_{ir})(1-\delta_{is})(1-\delta_{rs})\hmm_{i}\hmm_{r}\hmm_{s} \\
        + \delta_{is}(1-\delta_{jr})(1-\delta_{js})(1-\delta_{rs})\hmm_{j}\hmm_{r}\hmm_{s}.
    \end{multline*}
    Note that
    \begin{multline}\label{eqn: An trig VVq3}
        \alpha_{[ij]rs} = \left(\delta_{jr}-\delta_{js}\right)(1-\delta_{ir})(1-\delta_{is})\hmm_{i}\hmm_{r}\hmm_{s} \\
        + \left(\delta_{is}-\delta_{ir}\right)(1-\delta_{jr})(1-\delta_{js})\hmm_{j}\hmm_{r}\hmm_{s} .
    \end{multline}
    It then follows from formulas (\ref{eqn: An trig VVq1}), (\ref{eqn: An trig VVq2}) and (\ref{eqn: An trig VVq3}) that
    \begin{equation}\label{eqn.Antrig.VVA}
        A_{[ij]rs} = \hmM\hmm_{r}\hmm_{s}(\delta_{ir}\delta_{js}-\delta_{jr}\delta_{is})
        + \hmm_{r}\hmm_{s}\big(\hmm_{i}(\delta_{jr}-\delta_{js}) + \hmm_{j}(\delta_{is}-\delta_{ir})\big).
    \end{equation}
    \par
    We now consider those terms in $\lambda_{[ij]rs}$ which are linear in $\beta_{pq}$. By Lemmas \ref{thm: An trig FBF1} and \ref{thm: An trig Fk} the sum of these terms is given by $B_{[ij]rs}$, where
    \begin{multline*}
        B_{ijrs} = \sum_{k=1}^{n}\hmm_{k}\left(\delta_{kir}\sum_{q=1}^{n}\hmm_{q}\beta_{kq} + \delta_{ki}\hmm_{r}\beta_{rk} + \delta_{kr}\hmm_{i}\beta_{ik} + \delta_{ir}\hmm_{i}\beta_{ki}\right) \\
        \times \left(\delta_{kjs}\chm + \delta_{kj}\bhm\hmm_{s}+\delta_{ks}\bhm\hmm_{j}+\delta_{js}\bhm\hmm_{j}\right) \\
        + \sum_{k=1}^{n}\hmm_{k}\left(\delta_{kjs}\sum_{q=1}^{n}\hmm_{q}\beta_{kq} + \delta_{kj}\hmm_{s}\beta_{sk} + \delta_{ks}\hmm_{j}\beta_{jk} + \delta_{js}\hmm_{j}\beta_{kj}\right) \\
        \times\left(\delta_{kir}\chm + \delta_{ki}\bhm\hmm_{r}+\delta_{kr}\bhm\hmm_{i}+\delta_{ir}\bhm\hmm_{i}\right) .
    \end{multline*}

    By expanding, and cancelling terms which are symmetric in $i$ and $j$ or in $r$ and $s$, we see that
    \begin{equation*}
        B_{[ij]rs} = \Tilde{B}_{[ij][rs]},
    \end{equation*}
    where
    \begin{equation*}
        \Tilde{B}_{ijrs} = \hmm_{i}\hmm_{j}\Big(\delta_{is}\bhm\hmm_{r}\beta_{ri} + \delta_{jr}\bhm\hmm_{s}\beta_{ir} +\delta_{ir}\delta_{js}\chm\beta_{ji}+\delta_{ir}\bhm\hmm_{s}\beta_{ji}+\delta_{ir}\bhm\hmm_{s}\beta_{si}\Big),
    \end{equation*}
    and
    \begin{equation*}
        \Tilde{B}_{[ij][rs]} = \Tilde{B}_{[ij]rs} - \Tilde{B}_{[ij]sr} .
    \end{equation*}
    Note that the term $\delta_{ir}\delta_{js}\chm\hmm_{i}\hmm_{j}\beta_{ji}$ in $\Tilde{B}_{ijrs}$ cancels with the corresponding term in $\Tilde{B}_{jisr}$, and hence terms proportional to $\chm$ vanish in $\Tilde{B}_{[ij][rs]}$. The terms
    \begin{equation*}
        \delta_{is}\bhm\hmm_{i}\hmm_{j}\hmm_{r}\beta_{ri}+\delta_{ir}\bhm\hmm_{i}\hmm_{j}\hmm_{s}\beta_{si} = \bhm\hmm_{i}\hmm_{j}\beta_{rs}\left(\delta_{is}\hmm_{r}-\delta_{ir}\hmm_{s}\right)
    \end{equation*}
    are together symmetric under the swap of $r$ and $s$, hence the corresponding terms cancel in $\Tilde{B}_{ij[rs]}$. Finally, the sum of the remaining two terms is symmetric under the swap of $i$ and $j$, hence the corresponding terms cancel in $\Tilde{B}_{[ij]rs}$. Overall, we get that 
    \begin{equation}\label{eqn.Antrig.VVB}
        B_{[ij]rs}=\Tilde{B}_{[ij][rs]}=0 .
    \end{equation}
    \par
    Finally, we consider the contribution to $\lambda_{[ij]rs}$ from constant terms, denoted $C_{[ij]rs}$ with
    \begin{multline*}
        C_{ijrs} = \sum_{k=1}^{n}\hmm_{k}\left(\delta_{kir}\chm + \delta_{ki}\bhm\hmm_{r}+\delta_{kr}\bhm\hmm_{i}+\delta_{ir}\bhm\hmm_{i}\right)  \\
        \times \left(\delta_{kjs}\chm + \delta_{kj}\bhm\hmm_{s}+\delta_{ks}\bhm\hmm_{j}+\delta_{js}\bhm\hmm_{j}\right) .
    \end{multline*}
    We can again omit terms symmetric in $i$ and $j$ to find $C_{[ij]rs} = \Tilde{C}_{[ij]rs}$, where
    \begin{multline*}
        \Tilde{C}_{ijrs} = \bhm\hmm_{i}\hmm_{j}\big(\delta_{irs}\chm + \delta_{ir}\delta_{js}\left(2\chm+\bhm\hmM\right) + \delta_{is}\bhm\hmm_{r} + 2\delta_{js}\bhm\hmm_{r} + \delta_{jrs}\chm + \delta_{jr}\bhm\hmm_{s} \\
        + \delta_{rs}\bhm\hmm_{r} + 2\delta_{ir}\bhm\hmm_{s}\big).
    \end{multline*}
    Note that the sum of the terms
    \begin{equation*}
        \bhm\hmm_{i}\hmm_{j}\left(\delta_{irs}\chm + \delta_{jrs}\chm + \delta_{is}\bhm\hmm_{r} + \delta_{js}\bhm\hmm_{r} + \delta_{jr}\bhm\hmm_{s} + \delta_{ir}\bhm\hmm_{s} + \delta_{rs}\bhm\hmm_{r}\right)
    \end{equation*}
    is symmetric in $i$ and $j$, hence vanishes after the antisymmetrisation. 
    Therefore
    \begin{multline}\label{eqn.Antrig.VVC}
        C_{[ij]rs} = \Tilde{C}_{[ij]rs} = \bhm\hmm_{r}\hmm_{s}(\bhm\hmM+2\chm)(\delta_{ir}\delta_{js}-\delta_{jr}\delta_{is}) \\
        + \bhm^{2}\hmm_{j}\hmm_{r}\hmm_{s}(\delta_{ir}-\delta_{is}) + \bhm^{2}\hmm_{i}\hmm_{r}\hmm_{s}(\delta_{js}-\delta_{jr}) .
    \end{multline}
    \par
    To finish the proof, we combine expressions (\ref{eqn.Antrig.VVA})--(\ref{eqn.Antrig.VVC}) to obtain the required statement for
    \begin{equation*}
        \lambda_{[ij]rs} = \frac{1}{\bhm\hmM+\chm}\left(A_{[ij]rs} + B_{[ij]rs} + C_{[ij]rs}\right) .
    \end{equation*}
\end{proof}

We are now ready to prove Theorem \ref{thm: An trig}.
\begin{proof}[Proof of Theorem \ref{thm: An trig}]
    By Lemmas \ref{thm: An trig FBF1} and \ref{thm: An trig VV}
    \begin{multline}\label{eqn.Antrig.FBF0}
        \left(F_{i}\eta^{-1}F_{j}\right)_{rs}-\left(F_{j}\eta^{-1}F_{i}\right)_{rs} = \\
        \hmm_{r}\hmm_{s}(\delta_{jr}\delta_{is}-\delta_{ir}\delta_{js})\bigg(\frac{(\ahm\hmM+2\bhm)(\bhm\hmM+\chm)}{\ahm\hmM^{2}+3\bhm\hmM+\chm} - \frac{\hmM+\bhm^{2}\hmM+2\bhm\chm}{\bhm\hmM+\chm}\bigg) \\
        + \hmm_{r}\hmm_{s}\big(\hmm_{j}\left(\delta_{is}-\delta_{ir}\right)+\hmm_{i}(\delta_{jr}-\delta_{js})\big)\bigg(\frac{(\ahm\hmM+2\bhm)^{2}}{\ahm\hmM^{2}+3\bhm\hmM+\chm} - \ahm - \frac{\bhm^{2}-1}{\bhm\hmM+\chm}\bigg) \\
        = \frac{\hmm_{r}\hmm_{s}\gamma\big(\hmM(\delta_{ir}\delta_{js}-\delta_{jr}\delta_{is})+\hmm_{j}\left(\delta_{is}-\delta_{ir}\right)+\hmm_{i}(\delta_{jr}-\delta_{js})\big)}{\left(\ahm\hmM^{2}+3\bhm\hmM+\chm\right)\left(\bhm\hmM+\chm\right)} ,
    \end{multline}
    where
    \begin{equation*}
        \gamma = \bhm^{3}\hmM+3\bhm^{2}\chm-\ahm\chm^{2}+\ahm\hmM^{2}+3\bhm\hmM+\chm.
    \end{equation*}
    Therefore the WDVV equations hold when $\gamma = 0$ as stated. It is also easy to see that $\gamma$ has to be zero when $n \geq 3$ if expression (\ref{eqn.Antrig.FBF0}) vanishes for all $i,j,r,s$.
\end{proof}

\begin{rmk}
    In the special case $\bhm=-\frac{1}{2}\ahm\hmM$, the metric $\eta$ is diagonal, the condition (\ref{eqn: An trig cond2}) reduces to (\ref{eqn: An trig cond1}), and the relation (\ref{eqn: An trig condWDVV}) can be rearranged to the form
    \begin{equation*}
        \ahm^{2}\hmM^{2}-4\ahm\chm+4=0 .
    \end{equation*}
\end{rmk}

Theorem \ref{thm: An trig} relies on the two conditions (\ref{eqn: An trig cond1}) and (\ref{eqn: An trig cond2}).
We now consider the case when condition (\ref{eqn: An trig cond1}) does not hold.
The first statement is as follows.
\begin{thm}\label{thm.Antrig.sc1}
    Suppose $\bhm\hmM+\chm=0.$
    Let $Q$ be the diagonal $n\times n$ matrix with entries $Q_{rs}=\delta_{rs}\hmm_{r}$.
    Then the equations
    \begin{equation}\label{eqn: An trig notWDVV}
        F_{i}Q^{-1}F_{j}=F_{j}Q^{-1}F_{i}
    \end{equation}
    hold for all $i,j\in\{1,\dots,n\}$ if $b=\pm 1$.
    Moreover, equations (\ref{eqn: An trig notWDVV}) imply $b = \pm 1$ provided that $n\geq 3$.
\end{thm}
\begin{proof}
    Note that
    \begin{equation*}
        \left(F_{i}Q^{-1}F_{j}\right)_{rs} = \sum_{k=1}^{n}\frac{1}{\hmm_{k}}F_{irk}F_{jks} ,
    \end{equation*}
    and $F_{i}=W_{i}+V_{i}$ as in Lemma \ref{thm: An trig Fk}. By formula (\ref{eqn: An trig FBF}) and Lemmas \ref{thm: An trig FBF1}, \ref{thm: An trig VV}, we find that
    \begin{align*}
            \left(F_{i}Q^{-1}F_{j}\right)_{rs}-\left(F_{j}Q^{-1}F_{i}\right)_{rs} &= \sum_{k=1}^{n}\frac{1}{\hmm_{k}}\Big[\left(V_{i}\right)_{rk}\left(V_{j}\right)_{ks} - \left(V_{j}\right)_{rk}\left(V_{i}\right)_{ks}\Big] \\
            &= \left(\hmM+\bhm^{2}\hmM+2\bhm\chm\right)\hmm_{r}\hmm_{s}\left(\delta_{ir}\delta_{js}-\delta_{is}\delta_{jr}\right) \\
            & \; + \left(\bhm^{2}-1\right)\hmm_{r}\hmm_{s}\big(\hmm_{i}(\delta_{js}-\delta_{jr})+\hmm_{j}(\delta_{ir}-\delta_{is})\big) .
    \end{align*}
    Since $\chm=-\bhm\hmM$, we get $\hmM+\bhm^{2}\hmM+2\bhm\chm = \hmM(-\bhm^{2}+1)$ and the statement follows.
\end{proof}

We also have the following statement on a solution of the WDVV equations.
\begin{thm}\label{thm.Antrig.sc2}
    Suppose that $\bhm\hmM+\chm=0$, $b=\pm 1$, and $\ahm\hmM+2\bhm\neq 0$.
    Then the matrix $Q$ from Theorem \ref{thm.Antrig.sc1} can be represented as
    \begin{equation*}
        Q=\kappa^{-1}\sum_{k=1}^{n}h_{k}F_{k},
    \end{equation*}
    where
    \begin{equation}\label{eqn: An trig h}
        h_{k} = -(\ahm\hmM+2\bhm)e^{2\bhm\hmc_{k}} + \ahm\sum_{q=1}^{n}\hmm_{q}e^{2\bhm\hmc_{q}} 
    \end{equation}
    and
    \begin{equation*}
        \kappa = -2\bhm(\ahm\hmM+2\bhm)\sum_{q=1}^{n}\hmm_{q}e^{2\bhm\hmc_{q}} .
    \end{equation*}
\end{thm}
\begin{proof}
    By Lemma \ref{thm: An trig Fk} and substituting $\chm=-\bhm\hmM$, we have
    \begin{multline*}
        F_{krs} = \ahm\hmm_{k}\hmm_{r}\hmm_{s} + \delta_{krs}\hmm_{k}\left(\sum_{q=1}^{n}\hmm_{q}\beta_{kq} - \bhm\hmM\right) + \delta_{kr}\hmm_{k}\hmm_{s}\beta_{sk} \\
        + \delta_{ks}\hmm_{k}\hmm_{r}\beta_{rk} + \delta_{rs}\hmm_{k}\hmm_{r}\beta_{kr} + \bhm\hmm_{k}\left(\delta_{kr}\hmm_{s}+\delta_{ks}\hmm_{r}+\delta_{rs}\hmm_{r}\right).
    \end{multline*}
    We now compute the linear combination $B=\sum^{n}_{k=1}h_{k}F_{k}$, where $h_{k}$ is defined in (\ref{eqn: An trig h}).
    By expanding and cancelling terms, we find that the matrix entries of $B$ are
    \begin{multline}\label{eqn.Antrig.scB}
        B_{rs} = - \delta_{rs}(\ahm\hmM+2\bhm)\hmm_{r}\sum^{n}_{k=1}\hmm_{k}\beta_{rk}\left(e^{2\bhm\hmc_{r}}-e^{2\bhm\hmc_{k}}\right) + \delta_{rs}(\ahm\hmM+2\bhm)\bhm\hmm_{r}\hmM e^{2\bhm\hmc_{r}} \\
        + (\ahm\hmM+2\bhm)\hmm_{r}\hmm_{s}\beta_{rs}\left(e^{2\bhm\hmc_{r}}-e^{2\bhm\hmc_{s}}\right) - (\ahm\hmM+2\bhm)\bhm\hmm_{r}\hmm_{s}\left(e^{2\bhm\hmc_{r}}+e^{2\bhm\hmc_{s}}\right)\\
        - \delta_{rs}(\ahm\hmM+2\bhm)\bhm\hmm_{r}\sum^{n}_{k=1}\hmm_{k}e^{2\bhm\hmc_{k}} .
    \end{multline}
    Observe that
    \begin{equation*}
        \beta_{ij}\left(e^{2\bhm\hmc_{i}}-e^{2\bhm\hmc_{j}}\right) =
        \begin{cases}
            \bhm \left(e^{2\bhm\hmc_{i}}+e^{2\bhm\hmc_{j}}\right) &\textrm{ if } i\neq j, \\
            0 &\textrm{ otherwise.}
        \end{cases}
    \end{equation*}
    Then formula (\ref{eqn.Antrig.scB}) can be simplified to $B_{rs} = \kappa Q_{rs}$ as required.
\end{proof}
Theorems \ref{thm.Antrig.sc1} and \ref{thm.Antrig.sc2} generalise observations from the arXiv version of \cite{HM03-5D} to the case of arbitrary (non-zero) parameters $\hmm_{i}$.

\section{Comparing different trigonometric solutions}\label{sn.rat2trig}
\subsection{\texorpdfstring{$A_{n}$}{An}-type systems}\label{sn.Anr2t}
We show that an arbitrary choice of Legendre transform $S_{\gamma}$ applied to a rational solution $F^{\textrm{rat}}_{A_{n}(\dpa)}$ of the form (\ref{eqn: An rat sol}) produces a trigonometric solution $F^{\textrm{trig}}_{\hmm}$ of the form (\ref{eqn: An trig sol}), subject to some linear change of variables and a choice of parameters $\hmm_{i}, \ahm, \bhm, \chm$. We will deal with the most general case, where the deformation parameters $\dpa_{i}$ have arbitrary (non-zero) values. 
In the case when $\dpa_{i} = 1$ for all $i$, the solution $F^{\textrm{rat}}_{A_{n}(\dpa)}$ corresponds to the root system $A_{n}$. In this case, its Legendre transformation is equivalent to the solution found by Hoevenaars and Martini which is discussed in Section \ref{sn.Antrig}.
\par
Let $\widehat{F} = S_{\gamma}\left(F^{\textrm{rat}}_{A_{n}(\dpa)}\right)$ for some $\gamma \in \{1,\dots,n\}$, given explicitly in Theorem \ref{thm: An rat LT}. The coordinate system $\hat{x}$ is given by the Legendre transformation $S_{\gamma}$, see formulas (\ref{eqn: An rat hat cov}). We also set $\widetilde{F} = F^{\textrm{trig}}_{\hmm}$ given by formula (\ref{eqn: An trig sol}).
\begin{thm}
    Define $\ahm,\bhm,\chm,$ and $\hmm=(\hmm_{1},\dots,\hmm_{n})$ by
    \begin{equation}\label{eqn.Anrat2trig.params}
    \begin{aligned}
        \hmm_{\alpha} & = 
        \begin{cases}
            1 & \textrm{if } \alpha = \gamma , \\
            \dpa_{\alpha} & \textrm{otherwise;}
        \end{cases} \\
        \ahm & = -\frac{2}{\dpA+1} ; \\
        \bhm & = 1 ; \\
        \chm & = - (\dpA+\dpa_{\gamma}+1),
    \end{aligned}
    \end{equation}
    where $\dpA = \sum_{i=1}^{n}\dpa_{i}$.
    Then functions $\widehat{F}$ and $\widetilde{F}$ satisfy
    \begin{equation}\label{eqn.Anrat2trig.F2Frel}
        \frac{(\dpA+1)^{2}}{8\dpa_{\gamma}^{2}}\widehat{F}\left(\hat{x}\right) = \widetilde{F}\left(y\right) \textrm{ up to quadratic terms,}
    \end{equation}
    where the coordinate system $y$ is given by
    \begin{equation}\label{eqn: An rat2trig coord}
        y_{\alpha} =
        \begin{cases}
            -\frac{\dpA+1}{2\dpa_{\gamma}}\hat{x}^{\gamma} & \textrm{if } \alpha = \gamma , \\
            \frac{\dpA+1}{2\dpa_{\gamma}}\left(\hat{x}^{\alpha}-\hat{x}^{\gamma}\right) & \textrm{otherwise.}
        \end{cases}
    \end{equation}
\end{thm}
\begin{proof}
    Inverting the expressions in (\ref{eqn: An rat2trig coord}) gives us
    \begin{equation*}
        \hat{x}^{\alpha} = 
        \begin{cases}
            -\frac{2\dpa_{\gamma}}{\dpA+1}y_{\gamma} & \textrm{ if } \alpha = \gamma , \\
            \frac{2\dpa_{\gamma}}{\dpA+1}(y_{\alpha}-y_{\gamma}) & \textrm{ otherwise.}
        \end{cases}
    \end{equation*}
    We rearrange the sum of the terms in $\widehat{F}$ which are proportional to $\hat{x}^{i}\hat{x}^{j}\hat{x}^{l}$ as
    \begin{align*}
        -\sum_{i<j<l}\frac{2\dpa_{i}\dpa_{j}\dpa_{l}}{\dpa_{\gamma}}\hat{x}^{i}\hat{x}^{j}\hat{x}^{l} &= -2\hat{x}^{\gamma}\sum_{\substack{i < j\\ i,j \neq \gamma}}\dpa_{i}\dpa_{j}\hat{x}^{i}\hat{x}^{j} -\sum_{\substack{i<j<l\\ i,j,l\neq\gamma}}\frac{2\dpa_{i}\dpa_{j}\dpa_{l}}{\dpa_{\gamma}}\hat{x}^{i}\hat{x}^{j}\hat{x}^{l} .
    \end{align*}    
    We then make the coordinate transformation $\hat{x}\to y$ in the expression for $\widehat{F}$ to obtain
    \begin{align*}
        \widehat{F} & = \sum_{i\neq\gamma}\frac{8\dpa_{\gamma}^{2}\dpa_{i}}{(\dpA+1)^{2}}f(y_{i}-y_{\gamma}) + \sum_{\substack{i < j \\ i,j \neq \gamma}}\frac{8\dpa_{\gamma}^{2}\dpa_{i}\dpa_{j}}{(\dpA+1)^{2}}f(y_{i}-y_{j}) + \textrm{cubic terms} \\
        & = \frac{8\dpa_{\gamma}^{2}}{(\dpA+1)^{2}}\sum_{1\leq i<j\leq n}\hmm_{i}\hmm_{j}f(y_{i}-y_{j}) + \textrm{cubic terms} ,
    \end{align*}
    since the $\hmm_{\alpha}$ are defined by (\ref{eqn.Anrat2trig.params}).
    The function $\widetilde{F}$ can be rewritten as
    \begin{multline}\label{eqn: An rat2trig F1}
        \widetilde{F} = \sum_{1 \leq i < j \leq n}\hmm_{i}\hmm_{j}f\left(\hmc_{i}-\hmc_{j}\right) + \ahm\sum_{1 \leq i<j<k \leq n}\hmm_{i}\hmm_{j}\hmm_{k}y_{i}y_{j}y_{k} \\
        + \frac{1}{2}\sum^{n}_{i\neq j}(\ahm\hmm_{i}+\bhm)\hmm_{i}\hmm_{j}y_{i}^{2}y_{j} + \frac{1}{6}\sum_{i=1}^{n}(\ahm\hmm_{i}^{2}+3\bhm\hmm_{i}+\chm)\hmm_{i}y_{i}^{3}.
    \end{multline}
    We now compare coefficients of the various cubic terms in $\widetilde{F}$ and $\frac{(\dpA+1)^{2}}{8\dpa_{\gamma}^{2}}\widehat{F}$ to show that relation (\ref{eqn.Anrat2trig.F2Frel}) is satisfied.
    
     After the change of variables $\hat{x}\rightarrow y$ in $\widehat{F}$ and some algebraic manipulation, we find that the sum of the cubic terms of the form $y_{i}y_{j}y_{k}$ in $\frac{(\dpA+1)^{2}}{8\dpa_{\gamma}^{2}}\widehat{F}$, with $i,j,k$ distinct, is equal to
    \begin{equation*}
        -\frac{2}{\dpA+1}\sum_{\substack{i<j \\i,j\neq\gamma}}\dpa_{i}\dpa_{j}y_{\gamma}y_{i}y_{j} - \frac{2}{\dpA+1}\sum_{\substack{i<j<l \\i,j,l\neq\gamma}}\dpa_{i}\dpa_{j}\dpa_{l}y_{i}y_{j}y_{l}.
    \end{equation*}
    This is the same as the second term in (\ref{eqn: An rat2trig F1}), since $\hmm_{\gamma}=1$ and $\ahm = -\frac{2}{\dpA+1}.$
    \par
    Next, we consider terms of the form $y_{i}^{2}y_{j}$, with $i\neq j$, in $\frac{(\dpA+1)^{2}}{8\dpa_{\gamma}^{2}}\widehat{F}$.
    The coefficient of the sum of the terms proportional to $y_{\gamma}^{2}y_{i}$, for $i\neq\gamma$, in $\frac{(\dpA+1)^{2}}{8\dpa_{\gamma}^{2}}\widehat{F}$ is equal to
    \begin{multline}\label{eqn.Anrat2trig.comp1}
        \frac{\dpa_{\gamma}}{\dpA+1}\left[-\dpa_{i}\dpa_{\gamma}+2\dpa_{i}(\dpA+1-\dpa_{i})-2\dpa_{i}(\dpA-\dpa_{i}-\dpa_{\gamma})-\frac{2\dpa_{i}}{\dpa_{\gamma}}\sum_{\substack{j< l \\ j,l \neq i, \gamma}}\dpa_{j}\dpa_{l}\right] \\
        + \frac{\dpa_{i}}{2(\dpA+1)}\big[2\dpa_{i}(\dpA+1-\dpa_{i})-(\dpA+1)(\dpA+1+\dpa_{\gamma}-\dpa_{i})\big] \\
        + \frac{\dpa_{i}}{2(\dpA+1)}\sum_{j\neq i, \gamma}\dpa_{j}\left(3+3\dpA-4\dpa_{i}-2\dpa_{j}\right) .
    \end{multline}
    By replacing the first sum in this expression with
    \begin{equation*}
        -\frac{\dpa_{i}}{\dpa_{\gamma}}\left((\dpA-\dpa_{i}-\dpa_{\gamma})^{2}-\sum_{j\neq i, \gamma}\dpa_{j}^{2}\right)
    \end{equation*}
    and some straightforward but substantial manipulations, the expression (\ref{eqn.Anrat2trig.comp1}) can be simplified to $\frac{\dpa_{i}(\dpA-1)}{2(\dpA+1)}$.
    All terms of the form $y_{i}^{2}y_{j}$, with $i\neq j$, in $\frac{(\dpA+1)^{2}}{8\dpa_{\gamma}^{2}}\widehat{F}$ sum to
    \begin{equation*}
        y_{\gamma}^{2}\sum_{i\neq\gamma}\frac{\dpa_{i}(\dpA-1)}{2(\dpA+1)}y_{i} + y_{\gamma}\sum_{i\neq\gamma}\frac{\dpa_{i}(\dpA-2\dpa_{i}+1)}{2(\dpA+1)}y_{i}^{2} + \sum_{\substack{i\neq j \\i,j\neq\gamma}}\frac{\dpa_{i}\dpa_{j}(\dpA-2\dpa_{i}+1)}{2(\dpA+1)}y_{i}^{2}y_{j},
    \end{equation*}
    which is equal to the third term in (\ref{eqn: An rat2trig F1}) since we have set $\bhm=1$.

    In computing the coefficient of terms proportional to $y_{\gamma}^{3}$, it is useful to note that
    \begin{equation*}
        \sum_{\substack{i<j<l \\ i,j,l\neq\gamma}}\dpa_{i}\dpa_{j}\dpa_{l} = \frac{1}{6}(\dpA-\dpa_{\gamma})^{3}-\frac{1}{2}(\dpA-\dpa_{\gamma})\sum_{i\neq\gamma}\dpa_{i}^{2}+\frac{1}{3}\sum_{i\neq\gamma}\dpa_{i}^{3} .
    \end{equation*}
    The sum of all terms of the form $y_{i}^{3}$ in $\frac{(\dpA+1)^{2}}{8\dpa_{\gamma}^{2}}\widehat{F}$ can be expressed as
    \begin{equation*}
        \left(\frac{\dpA(1-\dpA)}{6(\dpA+1)}-\frac{\dpa_{\gamma}}{6}\right)y_{\gamma}^{3}  + \sum_{i\neq\gamma}\left(-\frac{\dpa_{i}^{3}}{3(\dpA+1)}+\frac{\dpa_{i}^{2}}{2}-\frac{\dpa_{i}(\dpA+\dpa_{\gamma}+1)}{6}\right)y_{i}^{3}.
    \end{equation*}
    This agrees with the last term in (\ref{eqn: An rat2trig F1}), as $\chm=-(\dpA+\dpa_{\gamma}+1)$.
\end{proof}

Note that for the choice of parameters in (\ref{eqn.Anrat2trig.params}) we have 
\begin{equation*}
    \bhm\hmM+\chm = \dpa_{\gamma} \neq 0
\end{equation*}
and
\begin{equation*}
    \ahm\hmM^{2}+3\bhm\hmM+\chm = -\frac{2\dpa_{\gamma}^{2}}{\dpA+1} \neq 0.
\end{equation*}
By Theorem \ref{thm: An trig}, it follows that condition (\ref{eqn: An trig condWDVV}) holds, which can also be checked directly.

We also note that in general the family of solutions represented by $\widetilde{F}$ is larger than the family given by $\widehat{F}$.

\subsection{\texorpdfstring{$B_{n}$}{Bn}-type systems}\label{sn.Bnr2t}
The function $\widehat{F} = S_{\gamma}\left(F^{\textrm{rat}}_{B_{n}(\dpb)}\right)$, given in equality (\ref{eqn: Bn rat LT}), has a similar form to a known trigonometric solution generated by a $\vee$-system of type $BC_{n-1}$. Following the results of Alkadhem and Feigin in \cite{AF20-TRIGVS}, specifically Theorem 5.5, the trigonometric $BC_{n-1}$-type WDVV solutions have the form
\begin{multline} \label{eqn: AFBCn}
    \widetilde{F}  = \frac{1}{3}\xi_{0}^{3}+h\xi_{0}\sum_{i=1}^{n-1}m_{i}\xi_{i}^{2}+\lambda r \sum_{i=1}^{n-1}m_{i}\Tilde{f}(\xi_{i}) \\
    + \lambda \sum_{i=1}^{n-1}\left(sm_{i}+\frac{1}{2}qm_{i}(m_{i}-1)\right)\Tilde{f}(2\xi_{i}) + \lambda q \sum_{1 \leq i < j \leq n-1}m_{i}m_{j}\Tilde{f}(\xi_{i}\pm\xi_{j}) ,
\end{multline}
with coordinates $\xi = (\xi_{0}, \xi_{1}, \dots , \xi_{n-1})\in \CC^{n}$, and independent parameters $q,r,s \in \CC$, such that $q\neq 0$, and $m=(m_{1}, \dots , m_{n-1}) \in \left(\CC^{\times}\right)^{n-1}$.
The constants $h$, $\lambda$, $M \in \CC^{\times}$ are defined as follows:
\begin{align}
    h & = r+4s+2q(M-1) ; \label{eqn.BCn.h}\\
    \lambda & = \left(\frac{2h^{3}}{q(r+8s+2q(M-2))} \right)^{1/2} ; \label{eqn.BCn.lambda}\\
    M & = \sum_{i=1}^{n-1}m_{i} ; \label{eqn.BCn.M}
\end{align}
with the requirement that 
\begin{equation}\label{eqn: BCn cond}
    q\left(r+8s+2q(M-2)\right) \neq 0 .
\end{equation}
Finally, the function $\Tilde{f}(z)$ is given by
\begin{equation}\label{eqn: ftilde def}
    \Tilde{f}(z) = -f(-iz),
\end{equation}
where $f(z)$ is defined in (\ref{eqn: f def}).
\par
First, we show that an arbitrary solution of the form $\widehat{F}$ is equivalent to a solution of the form $\widetilde{F}$.
\begin{thm}\label{thm: B2BC}
    Suppose $\dpB = \sum_{i=0}^{n}\dpb_{i} \neq \dpb_{\gamma}$.
    Define $q, r, s, m$ as
    \begin{align}
        q & = -\frac{2\dpb_{\gamma}^{2}}{R\dpB^{2}} ; \label{eqn.B2BC.q} \\
        r & = -\frac{4\dpb_{\gamma}^{2}(\dpb_{0}+\dpB)}{R\dpB^{2}} ; \label{eqn.B2BC.r} \\
        s & = \frac{\dpb_{\gamma}^{2}(\dpB-1)}{R\dpB^{2}} ; \label{eqn.B2BC.s} \\
        m_{\alpha} & =
        \begin{cases}
            \dpb_{\alpha} & \textrm{ if $1\le \alpha < \gamma$} ,\\
            \dpb_{\alpha+1} & \textrm{ if $\gamma \leq \alpha \leq n-1 $},
        \end{cases} \label{eqn.B2BC.m}
    \end{align}
    where $R\in\CC^{\times}$ is an independent scalar.
    Then functions $\widetilde{F}$ and $\widehat{F}$ satisfy
    \begin{equation}\label{eqn.Btrigrat.B2BCrel}
        \frac{R}{\lambda}\widetilde{F}\left(\xi\right) = \widehat{F}\left(\hat{x}\right)
    \end{equation}
    with coordinate transformation
    \begin{align}
        \xi_{0} & = \left(\frac{4\dpb_{\gamma}(\dpB-\dpb_{\gamma})}{R}\right)^{1/2} \hat{x}^{\gamma} ; \label{eqn.Btrigrat.B2BCy}\\
        \xi_{\alpha} & =
    \begin{cases}
       \frac{i\dpB}{\dpb_{\gamma}}\hat{x}^{\alpha} & \textrm{ if $1 \leq \alpha < \gamma$},\\
       \frac{i\dpB}{\dpb_{\gamma}}\hat{x}^{\alpha+1} & \textrm{ if $\gamma \leq \alpha \leq n-1 $} .
    \end{cases} \label{eqn.Btrigrat.B2BCxi}
    \end{align}
    Condition (\ref{eqn: BCn cond}) is satisfied by the choice of parameters (\ref{eqn.B2BC.q})--(\ref{eqn.B2BC.m}), and we also have $\lambda$, $h\neq 0$.
\end{thm}
\begin{proof}
    We start by inverting expressions (\ref{eqn.Btrigrat.B2BCy}) and (\ref{eqn.Btrigrat.B2BCxi}). The resulting coordinate transformation $\hat{x} \to \xi$ in $\widehat{F}$ gives us
    \begin{multline*}
        \widehat{F} = \sum_{i=1}^{n-1}\left[-\frac{4m_{i}\dpb_{\gamma}^{2}(\dpb_{0}+\dpB)}{\dpB^{2}}\Tilde{f}\left(\xi_{i}\right) + \frac{m_{i}\dpb_{\gamma}^{2}\left(\dpB-m_{i}\right)}{\dpB^{2}}\Tilde{f}\left(2\xi_{i}\right)\right] \\
        - \sum^{n-1}_{i < j}\frac{2m_{i}m_{j}\dpb_{\gamma}^{2}}{\dpB^{2}}\Tilde{f}\left(\xi_{i}\pm\xi_{j}\right) + \text{cubic terms} ,
    \end{multline*}
    where we have also used (\ref{eqn: ftilde def}) and (\ref{eqn.B2BC.m}).
    This expression matches the corresponding terms in $R\widetilde{F}/\lambda$. 
    
    Note that $M=\dpB-\dpb_{0}-\dpb_{\gamma}$ by (\ref{eqn.BCn.M}).
    Since we have $q,r,s, M$ in terms of the $\dpb_{\alpha}$, we may use (\ref{eqn.BCn.h}) and (\ref{eqn.BCn.lambda}) to derive
    \begin{align}
        h & = \frac{4\dpb_{\gamma}^{2}(\dpb_{\gamma}-\dpB)}{R\dpB^{2}} \label{eqn.Btrigrat.B2BCh}\\
        \intertext{and}
        \lambda^{2} & = \frac{16\dpb_{\gamma}(\dpB-\dpb_{\gamma})^{3}}{R\dpB^{2}} . \label{eqn.Btrigrat.B2BCl}
    \end{align}
    We can now consider the cubic terms in $\widehat{F}(\xi)$, which are
    \begin{equation*}
        \frac{\dpB R^{3/2}}{12\dpb_{\gamma}^{1/2}(\dpB-\dpb_{\gamma})^{3/2}}\xi_{0}^{3} - \frac{\dpb_{\gamma}^{3/2}R^{1/2}}{\dpB(\dpB-\dpb_{\gamma})^{1/2}}\xi_{0}\sum_{i=1}^{n-1}m_{i}\xi_{i}^{2} = \frac{R}{3\lambda}\xi_{0}^{3} + \frac{hR}{\lambda}\xi_{0}\sum_{i=1}^{n-1}m_{i}\xi_{i}^{2}.
    \end{equation*}
    This matches the cubic terms in $R\widetilde{F}/\lambda$: hence, relation (\ref{eqn.Btrigrat.B2BCrel}) holds.
    
    Substitution of the expressions for $q,r,s,M$ in terms of the $\dpb_{\alpha}$ into condition (\ref{eqn: BCn cond}) gives us
    \begin{equation*}
        q\left(r+8s+2q(M-2)\right) = -\frac{8\dpb_{\gamma}^{5}}{R^{2}\dpB^{4}} \neq 0 .
    \end{equation*}
    Also, expressions (\ref{eqn.Btrigrat.B2BCh}), (\ref{eqn.Btrigrat.B2BCl}) are non-zero.
\end{proof}

The family of solutions $\widehat{F}$ depends on $n+1$ parameters $\dpb_{0}, \dots, \dpb_{n}$ while the family $\widetilde{F}$ depends on $n+2$ parameters $q$, $r$, $s$, $m_{1}, \dots, m_{n-1}$. As we introduced an extra scale factor $R$ in Theorem \ref{thm: B2BC}, we can now relate a general solution of the family $\widetilde{F}$ to a solution of the family $\widehat{F}$. More precisely, the following statement holds.

\begin{thm}\label{thm: BC2B}
    Given a solution $\widetilde{F}$, such that $q\neq 2s$, we define:
    \begin{align}
        \dpb_{\alpha} & = 
        \begin{cases}
            \frac{r-2q+4s}{2q} & \textrm{ if $\alpha = 0$}, \\[0.5em]
            m_{\alpha} & \textrm{ if $0 < \alpha < \gamma$}, \\[0.5em]
            \frac{2q(2-M)-r-8s}{2q} & \textrm{ if $\alpha = \gamma$}, \\[0.5em]
            m_{\alpha -1} & \textrm{ if $\gamma < \alpha \leq n $};
        \end{cases} \label{eqn.BC2B.params}\\
        R & = - \frac{(r+8s+2q(M-2))^{2}}{2q(q-2s)^{2}} . \label{eqn.BC2B.R}
    \end{align}
    Then $\widehat{F}$ and $\widetilde{F}$ are related by $\frac{\lambda}{R}\widehat{F}(\hat{x}) = \widetilde{F}(\xi)$
    with coordinate transformation
    \begin{equation*}
        \hat{x}^{\alpha} = 
        \begin{cases}
            \frac{i(r+8s+2q(M-2))}{2(q-2s)}\xi_{\alpha} & \textrm{ if $1 \leq \alpha < \gamma$},\\[0.5em]
            \frac{h}{\lambda(q-2s)} \xi_{0} & \textrm{ if $\alpha = \gamma$}, \\[0.5em]
            \frac{i(r+8s+2q(M-2))}{2(q-2s)}\xi_{\alpha-1} & \textrm{ if $\gamma < \alpha \leq n$} .
        \end{cases}
    \end{equation*}
\end{thm}
\noindent The proof follows by inverting parameter relations and the change of coordinates from Theorem \ref{thm: B2BC}. In particular, it is useful to note that
\begin{equation*}
    \dpB = \dpb_{0} + \dpb_{\gamma} + M = \frac{q-2s}{q} .
\end{equation*}

\bigskip
{
  \noindent \textbf{Acknowledgments:}
  \par
  LK would like to thank the University of Glasgow for providing PhD funding, Dali Shen for discussions, and the organisers of the ISLAND 6 conference, where these results were first presented, for providing financial support.
}
\emergencystretch=1em
\raggedright
\printbibliography
\end{document}